\documentclass[12pt]{article}

\RequirePackage[OT1]{fontenc}
\RequirePackage{amsthm,amsmath,amsmath,amsfonts,amssymb,mathabx}
\RequirePackage{bm}
\RequirePackage{url}
\RequirePackage{natbib}
\usepackage[usenames]{color}
\RequirePackage[colorlinks,citecolor=blue,urlcolor=blue]{hyperref}
\usepackage{times}
\usepackage{appendix}
\usepackage{multirow}
\usepackage{natbib}
\usepackage{enumitem}
\usepackage{footmisc}
\usepackage{graphicx}
\usepackage{subfigure}
\usepackage{makecell}
\usepackage{booktabs}
\usepackage{xr}
\usepackage{booktabs}
\usepackage{array}
\usepackage{url}

\usepackage[ruled]{algorithm2e}
\usepackage{algorithmic}
\usepackage{dsfont}
\usepackage{authblk}
\usepackage{hyperref}

\allowdisplaybreaks[4]

\newtheorem{theorem}{Theorem}

\newtheorem*{lemma*}{Lemma}
\newtheorem{corollary}{Corollary}

\newtheorem{proposition}{Proposition}

\theoremstyle{definition}
\newtheorem{remark}{Remark}
\newtheorem{example}{Example}

\long\def\comment#1{}

\def\real{{\mathbb{R}}}

\newcommand{\red}{\color{red}}

\let\emptyset\varnothing

\renewcommand{\baselinestretch}{1.1}

\def\##1\#{\begin{align}#1\end{align}}
\def\$#1\${\begin{align*}#1\end{align*}}

\let\cite\citet

\let\hat\widehat
\let\tilde\widetilde

\def\given{{\,|\,}}

\newcommand{\xb}{\mathbf{x}}

\newcommand{\Xb}{\mathbf{X}}

\newcommand{\cA}{\mathcal{A}}

\newcommand{\cC}{\mathcal{C}}

\newcommand{\cE}{\mathcal{E}}

\newcommand{\cG}{\mathcal{G}}

\newcommand{\cI}{\mathcal{I}}

\newcommand{\cK}{\mathcal{K}}
\newcommand{\cL}{\mathcal{L}}
\newcommand{\cM}{\mathcal{M}}

\newcommand{\cO}{\mathcal{O}}

\newcommand{\cR}{\mathcal{R}}

\newcommand{\cT}{{\mathcal{T}}}

\newcommand{\cV}{\mathcal{V}}

\newcommand{\II}{\mathbb{I}}

\newcommand{\PP}{\mathbb{P}}

\newcommand{\RR}{\mathbb{R}}

\newcommand{\bSigma}{\bm{\Sigma}}

\newcommand{\argmin}{\mathop{\mathrm{argmin}}}

\DeclareMathOperator{\Var}{{\rm Var}}

\newcommand{\overbar}[1]{\mkern 1.5mu\overline{\mkern-1.5mu#1\mkern-1.5mu}\mkern 1.5mu}

\definecolor{DSgray}{cmyk}{0,1,0,0}

\makeatletter
\def\singlespace{\def\baselinestretch{1}\@normalsize}

\newcommand{\blind}{1}

\begin{document}
	
	\def\spacingset#1{\renewcommand{\baselinestretch}%
		{#1}\small\normalsize} \spacingset{1}

	\if1\blind
	{
		\title{ Ranking Inferences Based on the Top Choice of Multiway Comparisons}
		\author{  Jianqing Fan \qquad Zhipeng Lou  \qquad  Weichen Wang \qquad Mengxin Yu $^1$ }
		\date{}
		\maketitle
		\begin{singlespace}
			\begin{footnotetext}[1]
				{
					Jianqing Fan is Frederick L. Moore '18 Professor of Finance, Professor of Statistics, and Professor of Operations Research and Financial Engineering at Princeton University.  Weichen Wang is an Assistant Professor of Innovation and Information Management, Faculty of Business and Economics at The University of Hong Kong. Zhipeng Lou is a Postdoctoral Researcher at Department of Operations Research and Financial Engineering, Princeton University. Mengxin Yu is a Ph.D. student at Department of Operations Research and Financial Engineering, Princeton University, Princeton, NJ 08544, USA. 
					Emails: \texttt{\{jqfan, zlou,  mengxiny\}@princeton.edu} and \texttt{weichenw@hku.hk}.   Research supported by NSF grants DMS-2210833, DMS-2053832, DMS-2052926 and ONR grant N00014-22-1-2340				
				}
			\end{footnotetext}
			
		\end{singlespace}
	}
	\fi
	\if0\blind
	{
		\title{Ranking Inferences Based on the Top Choice of Multiway Comparisons}
		\author{}
		\date{}
		\maketitle
	} \fi

\begin{abstract}
This paper considers ranking inference of $n$ items based on the observed data on the top choice among $M$ randomly selected items at each trial.  This is a useful modification of the Plackett-Luce model for $M$-way ranking with only the top choice observed and is an extension of the celebrated Bradley-Terry-Luce model that corresponds to $M=2$.  Under a uniform sampling scheme in which any $M$ distinguished items are selected for comparisons with probability $p$ and the selected $M$ items are compared $L$ times with multinomial outcomes, we establish the statistical rates of convergence for underlying $n$ preference scores using both $\ell_2$-norm and $\ell_\infty$-norm, with the minimum sampling complexity.
In addition, we establish the asymptotic normality of the maximum likelihood estimator that allows us to construct confidence intervals for the underlying scores.  Furthermore, we propose a novel inference framework for ranking items through a sophisticated maximum pairwise difference statistic whose distribution is estimated via a valid Gaussian multiplier bootstrap.  The estimated distribution is then used to construct simultaneous confidence intervals for the differences in the preference scores and the ranks of individual items.  They also enable us to address various inference questions on the ranks of these items.  Extensive simulation studies lend further support to our theoretical results.  A real data application illustrates the usefulness of the proposed methods convincingly.
\end{abstract}


\newpage
\spacingset{1.45} 

\pagestyle{plain}

\section{Introduction}

The problem of ranking inference from pairwise comparisons or multiple partial rankings has drew significant attention in recent years, as the ranking problem has always played an important role in many applications such as individual choices in economics \citep{luce2012individual, mcfadden1973conditional}, psychology \citep{thurstone1927method, thurstone2017law}, online and offline recommendations \citep{baltrunas2010group, li2019estimating}, and ranking of items such as journals \citep{stigler1994citation, ji2022meta}, websites \citep{dwork2001rank}, colleges and universities \citep{avery2013revealed, caron2014bayesian}, sports teams \citep{massey1997statistical, turner2012bradley}, election candidates \citep{plackett1975analysis, mattei2013preflib}, and even alleles in genetics \citep{sham1995extended}. Previously the ranking problem has mostly focused on parameter estimation and algorithm implementation; see for example \cite{furnkranz2003pairwise, negahban2012iterative, azari2013generalized, maystre2015fast, jang2018top}. In addition, a large literature of empirical studies from the above areas of research focused on incorporating individual covariates for personalization \citep{turner2012bradley, li2019estimating}. However, the ranking inference has only received more attention in the statistics community recently. 

One of the most celebrated models for ranking problems is the Bradley-Terry-Luce (BTL) model. The model is frequently used to model pairwise comparisons. Specifically, consider a large collection of $n$ items whose true ranking is determined by some unobserved preference scores $\theta_{i}^*$ for $i=1,\dots,n$, for example, qualities of products, reputations of education institutes or abilities of sports teams. The BTL model assumes that an individual or a random event ranks item $i$ over $j$ with probability $\PP(\text{item i is preferred over j}) = e^{\theta_i^*} / (e^{\theta_i^*} + e^{\theta_j^*})$. The underlying choice axiom states that this probability does not depend on other items. Moreover, for simplicity, the BLT model does not account for data heterogeneity and treats the preference of each individual or event as purely independent Bernoulli. For the theoretical study, one may assume each pair $(i,j)$ is compared with probability $p$, and once compared, they are compared for $L$ times. Given the model and the collected data of pairwise comparisons, the statistical questions are straightforward: (a) What is the optimal statistical rate of convergence for estimating $\theta_i^*$ from the data? (b) What are the proper algorithms to achieve the optimal rate? (c) What is the asymptotic distribution of an estimator $\hat\theta_i$ of $\theta_i^*$? (d) Furthermore, how can we carry out uncertainty quantification on ranks? 

Questions (a) and (b) have been clearly addressed in \cite{negahban2016rank} for the $\ell_2$-loss of estimating $\pi^* = [\pi_1^*,\dots, \pi_n^*]^\top,$ where $\pi_i^* = e^{\theta_i^*} / \sum_i e^{\theta_i^*}$. They proposed the {\it rank centrality}, an efficient iterative spectral method. \cite{chen2019spectral} further studied the estimation of $\pi^*$ under the $\ell_{\infty}$-norm. They delivered the key message that both the spectral method and the regularized maximum likelihood estimator (MLE) can achieve the optimal statistical convergence rate under the $\ell_{\infty}$-norm and the sparsest possible sampling regime ($p \gtrsim \log n/n$). \cite{chen2020partial} further complements and refines the results of \cite{chen2019spectral} by concluding that the vanilla MLE without regularization can already achieve the optimal rate of convergence in both $\ell_2$- and $\ell_\infty$- norms for estimating $\theta^* = [\theta_1^*,\dots, \theta_n^*]^\top$ and the condition for exact recovery of the top-K ranking for MLE is weaker in constant than the spectral method.

Following (a) and (b), researchers also made recent progress on addressing (c) and (d)  for the BLT model. Specifically, \cite{han2020asymptotic} made contributions to show the asymptotic normality of the MLE estimator with the sampling regime of $p \gtrsim (\log n)^{1/5}/n^{1/10}$, while \cite{gao2021uncertainty} fully revealed the asymptotic normality of both the MLE estimator and the spectral estimator under the assumption that $p \gtrsim (\log n)^{1.5} / n,$ where $1.5$ may be further improved to $1+\delta$ for arbitrary $\delta > 0$. The authors showed that although the spectral method is optimal in the order of sample complexity, it is less efficient due to its larger asymptotic variance than the MLE. Despite a great theoretical contribution to the asymptotic normality, \cite{gao2021uncertainty} did not focus too much on the ranking inference problem (d) and only treated (d) with a crude confidence interval bound. In this work, we will refine the analysis of ranking inference for the MLE and therefore fill an important gap in the literature beyond the work of
\cite{gao2021uncertainty}. 
The ranking inference for the BLT model is also studied by \cite{liu2022lagrangian}, but with a rather strong assumption of $L \gtrsim n^2 \log^2 n $, that is, each compared pair must be compared more than $C n^2$ times, which is barely possible in practical applications. In contrast, we only require $L \gtrsim \mathrm{poly}(\log n)$ in this work to carry out our rank hypothesis testing.

Another more general model that extends pairwise comparison is the Plackett-Luce (PL) model, which assumes $M$-way full ranking. In the PL model, every time an individual provides a personal ranking on all given $M$ items. Denote this full ranking as $i_1 \succ \dots \succ i_M$. It can be understood as $M-1$ independent events that $i_1$ is preferred over the set $\{i_1, \dots, i_M\}$, $i_2$ is preferred over the set $\{i_2,\dots,i_M\}$ and so forth. The model assumes that
$$
    \PP(i_1 \succ \dots \succ i_M) = \prod_{j=1}^{M-1} \bigg[e^{\theta_{i_j}^*} / \sum_{k=j}^M e^{\theta_{i_k}^*}\bigg].
$$
Similar to the BLT model, each $M$-way comparison $\{i_1, \dots, i_M\}$ is compared with probability $p$, and once compared, they are ranked for $L$ times. In practice, $L$ can be different. Nevertheless, for simplicity, in this work, we assume $L$ is a shared quantity for each compared $M$ item to ease the presentation and computation. 
The set of all $M$-way comparisons forms a comparison hyper-graph, which we will formally define later.  
When $M=2$, the PL model reduces to the BLT model. Again we could ask the same four inference questions above for the PL model. 

The PL model has garnered less attention due to the more complicated structure of multiple comparisons, although it fits in with more general and real settings, including for instance, multi-player games and personal preferences with multiple products. Research papers on the inference problems based on the PL model are relatively scarce. \cite{maystre2015fast} introduced the iterative Luce spectral ranking method and showed that it converges to the MLE without providing any statistical rate. 
\cite{jang2018top} rigorously considered conditions for the exact recovery of top-K ranking and applied the spectral method to achieve this exact recovery under the assumption that $p \gtrsim (M-1)\sqrt{\log n/\binom{n-1}{M-1}}$. However, according to \cite{cooley2016threshold}, the sparsest regime that we can have a connected hyper-graph is when $p \gtrsim \log n/\binom{n}{M-1}$. In this paper, we close this gap by showing that we can achieve optimal estimation error based on the MLE under the sparsest regime with $p \gtrsim \textrm{poly}(\log n)/\binom{n-1}{M-1}$, even in the harder situation than the traditional PL model when only top choices are observed from the $M$-way comparisons. 
Moreover, it is worth noting that all aforementioned works  on $M$-way comparisons only focused on deriving first-order statistical rates of convergence, and the corresponding asymptotic distributions have rarely been investigated. To fill in this blank, we further derive the uncertainty quantification results in this sparest regime and apply them to study the practical ranking inference. 

More specifically, instead of working on the PL model, we consider the partial-ranking case in which we only observe the top choice from the choice set $\{i_1, \dots, i_M\}$. 
This is motivated from two perspectives. On the one hand, many applications do not provide the full ranking among all the $M$ selected items, and only the top choice is known. For example, a multi-player game may stop once we get the winner; a shopper may only choose the top item to purchase after presenting a set of products. On the other hand, theoretically speaking, general $M$ with full ranking gives a likelihood function that does not provide much more insight beyond only considering the likelihood for the top choice. The theory will be more concise and intuitive regarding the role of $M$ as we will see in later sections. For $M=3$, we will also present the asymptotic normality for the PL model when the full ranking of $M=3$ items is available. For the PL model with general $M > 3$, the MLE theory can be derived similarly, but due to its more tedious notational details, we decide to omit it.

Therefore our main focus of the paper is the multiway comparison model with only the top choice observed. Under this model, we apply the MLE method and analyze its statistical rates in both $\ell_2$- and $\ell_\infty$- norms for estimating $\theta^*$ and show that they are \emph{optimal} under the sparsest hyper-graph regime. In specific, when $M=\cO(1)$, we achieve the same statistical rate as the PL model presented in \cite{jang2018top} but only requiring $p\gtrsim \textrm{poly}(\log n)/\binom{n-1}{M-1}$, even if we only observe the top choice. This answers (a) and (b). Furthermore, to respond to the question in (c) for any $M \ge 2$, we establish the asymptotic normality of the MLE $\hat\theta_i$, for all $ i\in[n]$ in the sparest regime where the sampling probability satisfies $p\gtrsim \textrm{poly}(\log n)/\binom{n-1}{M-1}$. Finally, for question (d), we address three detailed ranking inference problems: (i) constructing valid confidence intervals for the ranks of a set of items, (ii) testing if an item belongs to the top-K ranked items, (iii) providing a sure screening confidence set that contains all the top-K ranked items with high confidence. All of these are important inference questions in practice. For example, when high school seniors choose their colleges,  they often care about the confidence interval for the ranks of a few universities, whether a certain university is within the top 50, and a list of universities that contain top 50 institutes with say a 95\% confidence level. In order to complete these tasks, based on the asymptotic normality of the MLE, we propose a novel inference framework for ranking items through a sophisticated maximum pairwise difference statistic whose distribution is estimated via a valid Gaussian multiplier bootstrap \citep{CCK2017, Chernozhukov2019}. 
The estimated distribution is then used to construct simultaneous confidence intervals for the differences in the preference scores and the ranks of individual items. They also
enable us to address the above inference questions on the ranks of these items. 

Our main contributions of the work are summarized as follows. Firstly, we study the performance of the MLE on the more complicated general multiway comparison model with only top choice observed and show that MLE can achieve the optimal sample complexity under the sparsest possible regime. Secondly, we quantify the uncertainty of the MLE explicitly. Last but not least, we provide a general framework to conduct effective inference of ranks based on the Gaussian multiplier bootstrap and give answers to three crucial practical inference questions. Specifically, our proposed confidence intervals constructed for individual ranks are provably narrower than the high confidence Bonferroni adjustment in \cite{gao2021uncertainty}.

\subsection{Roadmap}

In Section \ref{sec:algo}, we set up the model  with some basic assumptions. Section \ref{sec:ranking_estimation} presents the performance of parameter estimation  and asymptotic distribution of the MLE, while Section \ref{sec:ranking_inference} details newly the proposed framework for constructing rank confidence intervals, rank testing statistics, and top-K sure screening confidence set. Section \ref{sec:numerical} contains comprehensive numerical studies to verify theoretical results and a real data example to illustrate the usefulness of our ranking inference methods. Finally we conclude the paper with some discussions in Section \ref{sec:discussion}. All the proofs are deferred to the appendix.

\subsection{Notation}
Throughout this work, we use $[n]$ to denote the index set $\{1,2,\cdots,n\}.$ For any given vector $\xb\in \RR^{n}$ and $q\ge 0$, we use $\|\xb\|_{q}=(\sum_{i=1}^{n}|x_i|^q)^{1/q}$ to represent the vector $\ell_q$ norm.  For any given matrix $\Xb\in \RR^{d_1\times d_2}$, we use $\|\cdot\|$ to denote the spectral norm of $\Xb$ and  write $\Xb\succcurlyeq 0$ or $\Xb\preccurlyeq 0$ if $\Xb$ or $-\Xb$ is positive semidefinite.  For event $A$, $\II_{A}$ denotes an indicator random variable which equals $1$ if $A$ is true and $0$ otherwise. In addition, we let $\nabla L(\cdot),\nabla^2L(\cdot)$ be the gradient and Hessian of a loss function $L(\cdot)$. 
For two positive sequences $\{a_n\}_{n\ge 1}$, $\{b_n\}_{n\ge 1}$, we write $a_n=\cO(b_n)$ or $a_n\lesssim b_n$ if there exists a positive constant $C$ such that $a_n\le C\cdot b_n$ and we write $a_n=o(b_n)$ if $a_n/b_n\rightarrow 0$. Similarly we have $a_n=\Omega(b_n)$ or $a_n\gtrsim b_n$ if $a_n/b_n\ge c$ with some constant $c>0$. We use $a_n=\Theta(b_n)$ (or $a_n\asymp b_n$) if $a_n=\cO(b_n)$ and $a_n=\Omega(b_n)$.  Given $n$ items, we use $\theta_i^*$ to indicate the underlying preference score of the $i$-th item. 
Define $r : [n] \to [n]$ as the rank operator on the $n$ items which maps each item to its population rank based on the preference scores. We write the rank of the $i$-th item as $r_i$ or $r(i)$. By default, we consider ranking from the largest score to the smallest score.

\section{Multiway Comparison Model} \label{sec:algo}
We first introduce the formulation of the ranking problem for the multiway comparison model. The model consists of three key components. 

\begin{itemize}
\item \textbf{Preference scores}: For a given group of $n$ items,  they are associated $n$ preference scores 
\begin{align*}
    \theta^*=[\theta_1^*,\cdots,\theta_n^*]^\top, 
\end{align*}
which are assumed to fall within a range, 
\begin{align*}
    \theta_i^*\in[\theta_{L},\theta_{U}],\,\,\forall i\in[n].
\end{align*}
with the condition number 
$
    \kappa:=\theta_U-\theta_L. 
$
This paper considers the case where $\kappa$ is a fixed constant independent of $n$. This represents a more challenging scenario in which all items under comparison have preference scores in the same order. Otherwise, we could apply a simple screening to easily differentiate obvious winners or losers and redo the analysis within only items with scores of the same order. 

\item \textbf{Comparison hypergraph}: Let $\cG=(\cV,\cE)$ be a comparison hypergraph, where the vertex set $\cV=\{1,2,\cdots, n\}$ denotes the $n$ items of interest. $M$ different items $(i_1,\cdots,i_M)$ are compared if $(i_1,\cdots,i_M)$ falls within the edge set $\cE$ of a $M$-way hypergraph. We assume a hyper-edge connecting any set $(i_1,\cdots,i_M)$ of size $M$ with probability $p$. 
Let $A_{i_1\cdots i_{M}}$ take value 1 if item set  $(i_1,\cdots,i_{M})$ is compared and 0 otherwise. Then it is a sequence of realizations from independent Bernoulli random variables with parameter $p$ representing whether $(i_1,\cdots,i_{M})$ is connected in the hypergraph. When $M=2$, the hypergraph becomes the well-known Erdos-Renyi graph. 

\item \textbf{Multinomial Comparisons:} For each $(i_1,\cdots,i_M)\in \cE$, we observe $L$ independent comparisons among items in $\{i_1,\cdots,i_{M}\}$ and let $\{y_{i_1}^{(\ell)},\cdots,y^{(\ell)}_{i_M}\}$ be the $\ell$-th outcome of the  comparison. If the most preferred item is $i_k$, then $y_{i_k}^{(\ell)} = 1$ and the others are zero. Thus, for each $\ell \in [L]$, $\{y_{i_k}^{(\ell)},k\in[M]\}$ follows the multinomial distribution independently with probability
\begin{align*}\Big \{p_{i_k}=\frac{e^{\theta_{i_{k}}^*}}{\sum_{j=1}^{M}e^{\theta_{i_j}^*}},k\in[M]\Big\}. \end{align*}
Further define $y_{i_k}=\frac{1}{L}\sum_{l=1}^{L}y_{i_k}^{(\ell)}$. We also denote $y_{i_k}^{(\ell)}$ as $y_{i_k,(i_1,\cdots,i_{k-1},i_{k+1},\cdots,i_M)}^{(\ell)}$ when we need to emphasize that $i_k$ is preferred over the remaining items $\{i_1,\cdots,i_{k-1},i_{k+1},\cdots,i_M\}$. But we prefer the shorter notation when it is clear from the context on the comparison set. 

\end{itemize}

Throughout the paper, we consider $M=\cO(1)$. This assumption is trivially satisfied in many practical applications. For example, in a multi-player game, the number of teams or contestants who compete with each other in every game is typically a fixed number or has a fixed upper bound. A student who faces the selection of education programs may only get offers from a few institutes. In a recommendation system such as an online shopping platform, due to the limited space of a webpage, only a fixed number of items can be exhibited on the first page, and a shopper may seldom turn to the second page to make the purchasing decision. Moreover, in all these examples, we typically only have access to the most preferred item instead of knowing the full ranking of all $M$ items. Even when we observe the full ranking, it may not be trusted as much as the top preference due to the challenges to give full ranking of multiple items.

This paper aims to provide the statistical estimation and uncertainty quantification of the underlying scores of all items. More importantly, we study statistical inference for ranks, which is very much underdeveloped for the multiway comparison model. 

\section{Estimation and Uncertainty Quantification}\label{sec:ranking_estimation}
In this section, we utilize the MLE to derive an estimator for the underlying scores $\{\theta_i^*\}_{i=1}^{n}$ of $n$ items and establish the statistical convergence rates and asymptotic normality.

\subsection{Statistical Estimation}\label{sec:mle_con} 

The negative-log-likelihood function for our multiway comparison model is given by
\begin{align}
    \ell_n(\theta)=-\sum_{i_1\neq\cdots\neq i_M}A_{i_1\cdots i_M}\bigg[\sum_{k=1}^{M}y_{i_k}\log \bigg(\frac{e^{\theta_{i_k}}}{\sum_{j=1}^{M} e^{\theta_{i_j}} }\bigg)\bigg].\label{likelihood_main}
\end{align}
 Here the expression \eqref{likelihood_main} is mainly for the purpose of theoretical analysis. For computation of the MLE, the first summation is over all trials of multiway comparisons and the second sum has only one non-vanishing term.

From the above likelihood function, $\theta^*$ is only identifiable up to additive shift. We assume $\mathbf{1}^\top\theta^*=0$ for model identifiability. Thus, the parameter space for $\theta^{*}$ is the following $\Theta(\kappa)$ for some positive fixed constant $\kappa < \infty$, where 
\begin{align}\label{eq_bound_theta_kappa}
    \Theta(\kappa) = \left\{\theta \in \mathbb{R}^{n} : \max_{1\leq m\leq n} \theta_{m} - \min_{1\leq m\leq n} \theta_{m} \leq \kappa \enspace \mathrm{and} \enspace \mathbf{1}^{\top} \theta = 0\right\}.
\end{align}
Thus, the MLE is given by 
\begin{align}\label{estimator}
    \hat\theta =\argmin_{\mathbf{1}^\top \theta=0}\ell_n(\theta).
\end{align}
The next theorem gives the rate of convergence for $\hat\theta.$

\begin{theorem}\label{thm_consistency}
If $p \gtrsim \textrm{poly}(\log n)/\binom{n-1}{M-1}$, then the MLE defined in \eqref{estimator} satisfies
\begin{align}
     \|\hat\theta-\theta^{*}\|^2_{2}&\lesssim \frac{n}{\binom{n-1}{M-1}pL},\label{l_2_consist}\\
    \|\hat\theta-\theta^{*}\|^2_{\infty}&\lesssim \frac{\log n}{\binom{n-1}{M-1}pL}.\label{l_infty_consist}
\end{align}
\end{theorem}

Theorem \ref{thm_consistency} presents the $\ell_2$- and $\ell_{\infty}$- statistical convergence rates for $\hat\theta$ when one chooses the most preferred item among $M$ given items.  This coincides with the best rate one can hope for if we ignore the logarithmic term. To understand this from the information perspective, note that the parameter $\theta_i$ appears only in the comparisons when item $i$ is involved and the expected number of comparisons involving item $i$ is $\binom{n-1}{M-1}pL$ for all $i\in [n]$. 
Therefore, the best estimation error of $\theta_i$ we can achieve is $\cO(({\binom{n-1}{M-1}pL})^{-1/2})$ for all $i\in [n]$, which matches the obtained bound for $\|\hat\theta-\theta^*\|_{\infty}$ if we ignore the logarithmic term.  It is also worth mentioning that when $M=2$, our model reduces to the well-known BTL model. The $\ell_2$- and $\ell_{\infty}$- statistical rates also match those in estimating BTL model that are optimal up to logarithm terms \citep{chen2019spectral, chen2020partial}.

We hope to point out that when $M=\cO(1)$, our $\ell_2$- and $\ell_{\infty}$- statistical rates are identical to those in \cite{jang2018top} for the $M$-way comparisons in the PL model via the spectral method. This reveals that only picking the top item, instead of full ranking of all $M$ items, is sufficient to recover the underlying scores with same order of accuracy.  In addition, note that the hypergraph with edge size of $M$ items is connected with high probability when $p\gtrsim \log n/\binom{n}{M-1}$ by \cite{cooley2016threshold}; otherwise there will be isolated points and the corresponding items are never ranked. 
Our assumption on the sampling probability $p \gtrsim \textrm{poly}(\log n)/\binom{n-1}{M-1}$ matches the lower bound on sampling probability up to logarithmic terms, whereas \cite{jang2018top} requires $p \gtrsim (\log n/\binom{n-1}{M-1})^{1/2}$, a order of magnitude larger than ours.

The next corollary provides the conditions on the recovery of the top-$K$ items when there exists a gap between the scores of the true $K$-th and $(K+1)$-th items. 

\begin{corollary}\label{top-k-ranking} Under the conditions of Theorem \ref{thm_consistency}, if we have $\theta^*\in\Theta(\kappa)$ with $\theta^*_{(K)}-\theta^*_{(K+1)}\ge \Delta,$ we are able to recover the true top-$K$ items when the sample complexity satisfies 
\begin{align*}
\binom{n-1}{M-1}pL\gtrsim \Delta^{-2}\cdot \log n.
\end{align*}
Here $\theta_{(i)}^*$ denotes the underlying score of the item with true rank $i$ for $i\in[n]$.
\end{corollary}
This corollary follows directly from Theorem \ref{thm_consistency}. Note that when $M=2$, the requirement for sample complexity exactly reduces to that in \cite{chen2019spectral}, which is minimax optimal up to logarithm factors.
\subsection{Uncertainty Quantification}
This subsection aims at providing uncertainty quantification of estimator $\hat\theta$. We follow the idea proposed in \cite{gao2021uncertainty}, depicting the asymptotic behavior of every element of $\hat\theta$ via the likelihood function.

Before proceeding, we first separate out the likelihood terms involving the $m$-th entry $\theta_{m}$.  Fixing  other components $\theta_{-m}=\{\theta_i: i \ne m\}$, we define
\begin{align}\label{m_out_loss}
     \ell_n^{(m)}(\theta_m\given \theta_{-m})&=-M\sum_{(i_1\neq\cdots\neq i_{M-1})\neq m}A_{i_1\cdots i_{M-1}m}\bigg[\sum_{k=1}^{M-1}y_{i_k}\log \bigg(\frac{e^{\theta_{i_k}}}{\sum_{j=1}^{M-1} e^{\theta_{i_j}}+e^{\theta_{m}} }\bigg)\\&\qquad+y_{m}\log\bigg(\frac{e^{\theta_{m}}}{\sum_{j=1}^{M-1} e^{\theta_{i_j}}+e^{\theta_{m}} }\bigg)\bigg].\nonumber
\end{align}
Again this includes all terms in the likelihood function \eqref{likelihood_main} that has the information of $\theta_m$.  Let $f^{(m)}(\theta_m\given \theta_{-m})$ be the gradient of $\ell_n^{(m)}(\theta_m\given \theta_{-m})$ w.r.t. $\theta_{m}$, which is given by
\begin{align}
\label{f_m_main}
    f^{(m)}(\theta_m\given \theta_{-m})=M\sum_{(i_1\neq\cdots\neq i_{M-1})\neq m}A_{i_1\cdots i_{M-1}m}\bigg\{\frac{e^{\theta_m}}{\sum_{j=1}^{M-1}e^{\theta_{i_j}}+e^{\theta_m}}-y_{m} \bigg\}.
\end{align}
In addition, we also define  $g^{(m)}(\theta_m\given \theta_{-m})$ as the second derivative of $\ell_n^{(m)}(\theta_m\given \theta_{-m})$ w.r.t. $\theta_{m}$,  which is given by
\begin{align}
    g^{(m)}(\theta_m\given \theta_{-m})=M\sum_{(i_1\neq\cdots\neq i_{M-1})\neq m}A_{i_1\cdots i_{M-1}m}\bigg\{\sum_{j=1}^{M-1}\frac{e^{\theta_m+\theta_{i_j}}}{(\sum_{j=1}^{M-1}e^{\theta_{i_j}}+e^{\theta_m})^2} \bigg\}.\label{g_m_main}
\end{align}
Then, to maximize the likelihood, $\hat\theta_m$ must be the minimzer of $\ell_{n}^{(m)}(\theta\given \hat\theta_{-m})$.
By Taylor expansion of $\ell_{n}^{(m)}(\theta\given \hat\theta_{-m})$, a good proxy to $\hat\theta_{m}-\theta^{*}_m $ is its score function $-\frac{f^{(m)}(\theta_m^{*}\given \hat\theta_{-m})}{g^{(m)}(\theta_m^{*}\given \hat\theta_{-m})}$, which is approximately the same as $-\frac{f^{(m)}(\theta_m^{*}\given \theta_{-m}^{*})}{g^{(m)}(\theta_m^{*}\given \theta_{-m}^{*})}$. 
This leads us to consider the heuristic expression
\begin{align}\label{mle_decomp}
    \hat\theta_{m}-\theta^{*}_m=-\frac{f^{(m)}(\theta_m^{*}\given \theta_{-m}^{*})}{g^{(m)}(\theta_m^{*}\given \theta_{-m}^{*})} +\delta_m,
\end{align}
where we expect $\delta_m$ to be of smaller order.
The asymptotic distribution of $\hat\theta_{m}-\theta_m^{*}$ will then follow from 
that of $-\frac{f^{(m)}(\theta_m^{*}\given \theta_{-m}^{*})}{g^{(m)}(\theta_m^{*}\given \theta_{-m}^{*})}$. The following theorem makes the above heuristic discussion rigorous.

\begin{theorem}\label{thm_inference_k3}
If $p \gtrsim \textrm{poly}(\log n)/\binom{n-1}{M-1}$,  the MLE defined in \eqref{estimator} enjoys
\begin{align*}
\hat\theta_{m}-\theta^{*}_m=-\frac{f^{(m)}(\theta_m^{*}\given \theta_{-m}^{*})}{g^{(m)}(\theta_m^{*}\given \theta_{-m}^{*})} +\delta_m,
    \end{align*}
for all $m\in[n]$    with $
    \|\delta\|_{\infty}=o(({1}/{\binom{n-1}{M-1}pL})^{1/2})$ where $\delta=(\delta_1,\cdots,\delta_n)$. 
In addition, 
\begin{align*}
   \rho_m(\theta) (\hat\theta_m-\theta_m^*)\rightarrow N(0,1),
\end{align*}
for all $m\in[n]$ with
\begin{align}\label{eq_rho_theta}
 \rho_m(\theta)&=\bigg[\frac{L}{(M-1)!}\sum_{(i_1\neq\cdots\neq i_{M-1})\neq m}A_{i_1\cdots i_{M-1}m}\bigg\{\sum_{j=1}^{M-1}\frac{e^{\theta_m+\theta_{i_j}}}{(\sum_{j=1}^{M-1}e^{\theta_{i_j}}+e^{\theta_m})^2} \bigg\}\bigg]^{1/2},
\end{align}
for both $\theta\in\{\theta^*,\hat\theta\}.$
\end{theorem}

Theorem \ref{thm_inference_k3} presents the asymptotic distribution for every element of $\hat\theta$, by approximating $\hat\theta_{m}-\theta^{*}_m$ by the score function $-\frac{f^{(m)}(\theta_m^{*}\given \theta_{-m}^{*})}{g^{(m)}(\theta_m^{*}\given \theta_{-m}^{*})}$. Although a series of works have studied the ranking estimation problem under pairwise or general multiway comparison \citep{maystre2015fast,jang2018top}, results on uncertainty quantification are still very much underdeveloped. 
To our best knowledge, existing literature only has results on quantifying the uncertainty of the BTL model \citep{gao2021uncertainty}. We take one step further to unravel the uncertainty of the preference score estimator $\hat\theta$ presented in Theorem \ref{thm_inference_k3} for the more general $M$-way comparison with $M \ge 2$. 

We comment on the connections of our results with several related literature. Firstly, compared with \cite{liu2022lagrangian} who study the asymptotic distribution of their estimator via the Lagrangian debiasing method, we have no requirement on the number of comparisons $L$ for establishing the asymptotic distribution whereas they require $L \gtrsim n^2$. Secondly, compared with \cite{han2020asymptotic}, who established the asymptotic results in the regime when $p\gtrsim 1/n^{1/10}$, we allow much sparser regime for the comparison hypergraph, namely, we allow  $p \gtrsim \textrm{poly}(\log n)/\binom{n-1}{M-1}$. This matches the sparsest sampling regime up to logarithm terms. In addition, when $M=2$, our model reduces to the BTL model whose uncertainty quantification is provided by \cite{gao2021uncertainty} with the requirement of $p\gtrsim \textrm{poly}(\log n) /n$. 

As we mentioned in Section~\ref{sec:mle_con}, the statistical rate achieved under our top choice multiway comparison model is the same as that in estimating the conventional PL model up to logarithm terms. Hence the difference in the asymptotic behaviors under these two models lies in their different asymptotic variances. In specific, for any fixed $M$, when the underlying scores $\{\theta_{i}^*\}_{i=1}^n$ are of the same order, we sacrifice a factor of order $M$ in our asymptotic variance compared with using the PL model. 
We want to emphasize that our analyzing techniques can be easily extended to the scenario where one ranks the top-$k$ items for any $k\le M$ given the $M$ items. This includes the PL model as a particular case with $k=M$. 
We present a formal theorem on the asymptotic normality of the MLE for the PL model with $M=3$ in Appendix {\red B.8}. The corresponding theorem for the PL model with general $M=\cO(1)$ can be derived similarly, and we leave the details to the interested readers. As we argued, we often only have access to top choices in reality. More importantly, studying the top choice in multiway comparison already conveys our key messages clearly without excessive technicality.

\section{Ranking Inferences}
\label{sec:ranking_inference}

In many real applications, people have access to ranking-related data and problems. For instance, multiple sources such as US News and Times Higher Education publish global university rankings every year; sports team rankings are an essential part of our everyday chat; companies try to hire the best candidates based on the evaluations of their interviewers. Most current practical usage of ranks only involves estimating preference scores and displaying the estimated ranks. However, we often lack the tools to address basic inference questions such as the following.
\begin{itemize}
    \item Is Team A indeed significantly stronger than Team B? Can one build an efficient confidence interval for the ranks of a few items of interest?
    \item Is the offer from a university truly a good choice to accept? How can we tell whether an item is among the top-$K$ ranking with high confidence? 
    \item How many candidates should a company hire to ensure all the best candidates are selected? How do we get a confidence set of items to ensure the screening of the top-$K$ items?
\end{itemize}
In this section, we hope to address all these critical statistical inference questions for ranks. 
To this end, we introduce a novel inference framework for the population ranks  $\{r_{m}\}_{m \in \cM}$ simultaneously, where $\cM$ is any subset of $[n]$ of interest and $r_m,m\in[n]$ is the true rank in descending order of the $m$-th item according to its underlying preference score $\theta_m^*$. 

\subsection{Two-Sided Confidence Intervals}
In this subsection, we propose a general framework for constructing two-sided confidence intervals for ranks based on the MLE estimator $\hat{\theta} = (\hat{\theta}_{1}, \ldots, \hat{\theta}_{n})^{\top}$ given in~\eqref{estimator}.
To construct the simultaneous confidence intervals for the ranks, a direct approach is to derive the asymptotic distribution of the corresponding empirical ranks $\{\hat{r}_{m}\}_{m \in \cM}$ and figure out the critical value. However, as $\hat{r}_{m}$ is an integer which depends on all $\hat{\theta}_{1}, \ldots, \hat{\theta}_{n}$ for any $m \in \cM$, this is nontrivial. 

By exploiting the mutual relationship between the scores and the ranks, we observe that constructing confidence intervals for the ranks can be reduced to constructing simultaneous confidence intervals for the pairwise differences between the population scores, whose empirical counterpart's distribution is easier to depict. Therefore, to circumvent the difficulty of deriving the distribution of ranks directly, we work on the statistics of estimated scores $\{\hat{\theta}_{m}\}_{m \in [n]}$ instead. 
Next, we will illustrate the key intuition of our approach via the following example.

\begin{example}[Simultaneous rank confidence intervals]
\label{Example_confidence_interval}
Let $\cM = \{m\}$ for some $1\leq m\leq n$ be the item of interest. We consider constructing the $(1 - \alpha) \times 100\%$ confidence interval for the population rank $r_{m}$, where $\alpha \in (0, 1)$ is a prescribed significance level. Let $\{[\cC_{L}(k, m), \cC_{U}(k, m)]\}_{k \neq m}$ denote the simultaneous confidence intervals of $\{\theta_{k}^{*} - \theta_{m}^{*}\}_{k \neq m}$ such that with probability at least $1 - \alpha$, we have all $\theta_{k}^{*} - \theta_{m}^{*} \in [\cC_{L}(k, m), \cC_{U}(k, m)]$. Observe that $\cC_{U}(k, m) < 0$ (resp. $\cC_{L}(k, m) > 0$) implies $\theta_{k}^{*} < \theta_{m}^{*}$ (resp. $\theta_{k}^{*} > \theta_{m}^{*}$). Counting the number of items ranked lower than item $m$ by using confidence upper bounds and the number of items ranked above item $m$ by using the confidence lower bounds, we can get a confidence interval for $r_m$.  In other words, 
\begin{align}\label{eq_confidence_interval_m_probability}
    \PP\Bigg(1 + \sum_{k \neq m} \mathbb{I}\{\cC_{L}(k, m) > 0\} \leq r_{m} \leq n - \sum_{k \neq m} \mathbb{I}\{\cC_{U}(k, m) < 0\}\Bigg) \geq 1 - \alpha. 
\end{align}
This yields a $(1 - \alpha) \times 100\%$ confidence interval for $r_{m}$.
\end{example}

We now formally introduce the procedure to construct the confidence intervals for multiple ranks $\{r_{m}\}_{m \in \cM}$ simultaneously. Motivated by Example~\ref{Example_confidence_interval}, the key step is to construct the simultaneous confidence intervals for the pairwise score differences $\{\theta_{k}^{*} - \theta_{m}^{*}\}_{m \in \cM, k \neq m}$. Towards this end, define 
\begin{align}\label{statistics}
    \mathcal{T}_{\eta} = \max_{m \in \mathcal{M}} \max_{k \neq m} \left|\frac{\sqrt{L}\{\hat{\theta}_{k} - \hat{\theta}_{m} - (\theta_{k}^{*} - \theta_{m}^{*})\}}{\eta_{mk}}\right|. 
\end{align}
Here $\{\eta_{mk}\}_{1\leq k\neq m\leq n}$ is a sequence of positive normalizing constants introduced to account for different scales of $\{\hat{\theta}_{k} - \hat{\theta}_{m} - (\theta_{k}^{*} - \theta_{m}^{*})\}_{1\leq k\neq m\leq n}$. A natural choice of $\{\eta_{mk}\}_{\{1\leq k\neq m\leq n\}}$ is the uniform consistent estimators of the standard deviations $\{\hat{\sigma}_{mk}\}$ of $\{\hat{\theta}_{k} - \hat{\theta}_{m}\}_{\{1\leq k\neq m\leq n\}}$, where
\begin{align}\label{eq_sigma_hat}
    \hat{\sigma}_{mk}^{2} = \frac{M(M - 1)!}{g^{(m)}(\hat{\theta}_{m}|\hat{\theta}_{-m})} + \frac{M(M - 1)!}{g^{(k)}(\hat{\theta}_{k}|\hat{\theta}_{-k})}.
\end{align}

For any $\alpha \in (0, 1)$, let $\zeta_{1 - \alpha}$ denote a consistent estimate of the $(1 - \alpha)$ quantile of the asymptotic distribution of $\cT_{\eta}$ such that
\begin{align}\label{eq_zeta_asymptotic}
    \PP(\cT_{\eta} \leq \zeta_{1 - \alpha}) \to 1 - \alpha.
\end{align}
Then, motivated by Example~\ref{Example_confidence_interval}, the $(1 - \alpha)\times 100\%$ simultaneous confidence intervals for $\{r_{m}\}_{m \in \cM}$ are given by $\{[\cR_{m}^{\diamond}(\eta, \zeta_{1 - \alpha}), \cR_{m}^{\sharp}(\eta, \zeta_{1 - \alpha})]\}_{m \in \cM}$, where  
\begin{align}\label{eq_r_m_construction}
    \cR_{m}^{\diamond}(\eta, \zeta_{1 - \alpha}) &= 1 + \sum_{k \neq m} \mathbb{I}\left\{\hat{\theta}_{k} - \hat{\theta}_{m} > \frac{\eta_{mk}}{\sqrt{L}}\times \zeta_{1 - \alpha}\right\},\cr
    \cR_{m}^{\sharp}(\eta, \zeta_{1 - \alpha}) &= n - \sum_{k \neq m} \mathbb{I}\left\{\hat{\theta}_{k} - \hat{\theta}_{m} < -\frac{\eta_{mk}}{\sqrt{L}}\times \zeta_{1 - \alpha}\right\}.
\end{align}
In view of~\eqref{eq_confidence_interval_m_probability} and \eqref{eq_zeta_asymptotic}, the constructed simultaneous confidence intervals satisfy that 
\begin{align*}
    \PP\left(\cap_{m \in \cM}\{\cR_{m}^{\diamond}(\eta, \zeta_{1 - \alpha}) \leq r_{m} \leq \cR_{m}^{\sharp}(\eta, \zeta_{1 - \alpha})\}\right) \geq 1 - \alpha - o(1). 
\end{align*}

As noted in \eqref{eq_zeta_asymptotic}, the key step for constructing the confidence interval of ranks of interest is to pick the critical value $\zeta_{1 - \alpha}$. In the next subsection, we propose to estimate $\zeta_{1 - \alpha}$ via the Gaussian multiplier bootstrap procedure. 

\subsection{Gaussian Multiplier Bootstrap}
\label{Section_GMB}
We now present our proposed estimate of the critical value $\zeta_{1 - \alpha}$ for $\cT_{\eta}$, which satisfies~\eqref{eq_zeta_asymptotic}. We begin by introducing some notations and definitions. For each $1\leq \ell \leq L$, define  
\begin{align}
\label{eq_definition_xi_ml}
    \xi_{m\ell} = \frac{M}{g^{(m)}(\theta_{m}^{*}|\theta_{-m}^{*})} \sum_{(i_{1} \neq \cdots \neq i_{M - 1}) \neq m} A_{i_{1}\cdots i_{M - 1}m} \left\{\frac{e^{\theta_{m}^{*}}}{\sum_{j = 1}^{M - 1} e^{\theta_{i_{j}}^{*}} + e^{\theta_{m}^{*}}} - y_{m}^{(\ell)}\right\}.
\end{align} 
In addition, let $\xi_{m} := \frac{f^{(m)}(\theta_{m}^{*}|\theta_{-m}^{*})}{g^{(m)}(\theta_{m}^{*}|\theta_{-m}^{*})} = \frac{1}{L}\sum_{\ell = 1}^{L} \xi_{m\ell}$ for each $m \in [n]$. Theorem~\ref{thm_inference_k3} ensures that $\hat{\theta}_{m} - \hat{\theta}_{k} - (\theta_{m}^{*} - \theta_{k}^{*}) \approx \frac{1}{L} \sum_{\ell = 1}^{L} (\xi_{k\ell} - \xi_{m\ell})$ uniformly for all $1\leq k\neq m \leq n$. Consequently,  we obtain
\begin{align}
\label{eq_linear_Approximation_T}
    \cT_{\eta} \approx \max_{m \in \cM} \max_{k \neq m} \frac{1}{\sqrt{L}} \left|\sum_{\ell = 1}^{L} \left(\frac{\xi_{k\ell} - \xi_{m\ell}}{\eta_{mk}}\right)\right|. 
\end{align}
Notice that $\{\xi_{k\ell} - \xi_{m\ell}\}_{1\leq \ell \leq L}$ is a sequence of i.i.d.~zero-mean random variables with conditional variance $\sigma_{mk}^{2} = \mathrm{Var}(\xi_{k\ell} - \xi_{m\ell}|A)$  where $\textrm{Var}(\cdot\given A)$ denotes the conditional variance given a fixed comparison hypergraph, i.e. 
\begin{align}\label{random_a}
    A=\{A_{i_1\cdots i_M}\} \textrm{ where recall } A_{i_1\cdots i_M}=1 \textrm{ if } (i_1,\cdots,i_M)\in\cE.
\end{align}
In what follows, we also write $\cT = \cT_{\eta}$ for simplicity.

Since the dimension of the random vectors $\{(\xi_{k\ell} - \xi_{m\ell})_{m \in \cM, k \neq m}\}_{1\leq \ell \leq L}$ is $(n-1)|\cM|$ and is usually much larger than $L$ (we only require $\textrm{poly}(\log n)/L \to 0$ in Theorem~\ref{Theorem_Bootstrap_consistency_simultaneous_inference} below), the classical multivariate central limit theorem cannot be utilized here to derive the asymptotic distribution of $\cT$. Instead, we shall invoke the high dimensional Gaussian approximation result~\citep{CCK2017, Chernozhukov2019} which quantifies the distance between the distribution functions of $\cT$ and its Gaussian analogue. Nevertheless, the asymptotic distribution of $\cT$ still depends on the unknown population scores $\{\theta_{m}^{*}\}_{m \in [n]}$ and the covariance structure of the random vector $(\xi_{m\ell} - \xi_{k\ell})_{m \in \cM, k \neq m}$. To approximate the asymptotic distribution of $\cT$, we apply a practically feasible Gaussian multiplier bootstrap procedure~\citep{CCK2017, Chernozhukov2019}. First define the empirical version of $\xi_{m\ell}$ as follows,
\begin{align*}
    \hat{\xi}_{m\ell} = \frac{M}{g^{(m)}(\hat{\theta}_{m}|\hat{\theta}_{-m})} \sum_{(i_{1} \neq \cdots \neq i_{M - 1}) \neq m} A_{i_{1}\cdots i_{M - 1}m} \left\{\frac{e^{\hat{\theta}_{m}}}{\sum_{j = 1}^{M - 1} e^{\hat{\theta}_{i_{j}}} + e^{\hat{\theta}_{m}}} - y_{m}^{(l)}\right\}.
\end{align*}
Let $\omega_{1}, \ldots, \omega_{L} \in \RR$ be i.i.d.~standard normal random variables. Then, in view of~\eqref{eq_linear_Approximation_T}, the bootstrap counterpart of $\cT$ is defined as  
\begin{align*}
    \mathcal{G} = \max_{m \in \mathcal{M}} \max_{k\neq m} \frac{1}{\sqrt{L}} \left|\sum_{\ell = 1}^{L} \left(\frac{\hat{\xi}_{k\ell} - \hat{\xi}_{m\ell}}{\hat{\sigma}_{mk}}\right)\omega_{\ell}\right|.
\end{align*}
For any $\alpha \in (0, 1)$, let $\cG_{1 - \alpha}$ denote the $(1 - \alpha)$-th quantile of $\mathcal{G}$, that is, 
\begin{align}\label{cutoff_value}
    \mathcal{G}_{1 - \alpha} = \inf\{z \in \mathbb{R} : \mathbb{P}(\mathcal{G} \leq z|Y, A) \geq 1 - \alpha\}. 
\end{align}
Here $\PP(\cdot\given Y,A)$ denotes the conditional probability where all randomness from $Y=\{y_{i_k}^{(\ell)}: i_k \in \{i_1,\cdots,i_M\},\ell\in[L]\}$ and $A$ defined in \eqref{random_a} is fixed.

\begin{theorem}
\label{Theorem_Bootstrap_consistency_simultaneous_inference}
Assume $\textrm{poly}(\log n)/L \to 0$. Then, under the conditions of Theorem~\ref{thm_inference_k3}, we have  
\begin{align}\label{eq_Bootstrap_consistency_simultaneous_inference}
    \left|\mathbb{P}\left\{\max_{m \in \mathcal{M}} \max_{k \neq m} \left|\frac{\sqrt{L}\{\hat{\theta}_{k} - \hat{\theta}_{m} - (\theta_{k}^{*} - \theta_{m}^{*})\}}{\hat{\sigma}_{mk}}\right| > \mathcal{G}_{1 - \alpha}\right\} - \alpha\right| \to 0. 
\end{align}
\end{theorem}

Theorem~\ref{Theorem_Bootstrap_consistency_simultaneous_inference} indicates that the estimated critical value $\mathcal{G}_{1 - \alpha}$ from the Gaussian multiplier bootstrap indeed controls the significance level of the simultaneous confidence intervals for ranks in $\mathcal{M}$ to the prespecified level $\alpha$. To make the inference valid, we do require $L$ to grow faster than $\textrm{poly}(\log n)$. So our current proposal does not work for say $L=1$, where each comparison set is only ranked by one single individual. Fortunately, this requirement is not too restrictive and can be satisfied by many practical datasets such as the more than 30 datasets on elections, Netflix movie ranking, and sports competitions hosted on the PrefLib website \citep{mattei2013preflib}.

\begin{remark}\label{remark4.1}
Recall the definition of $\rho_{m}(\theta)$ in~\eqref{eq_rho_theta} for each $m \in [n]$. In the context of pairwise comparison where $M = 2$, \citet{gao2021uncertainty} proposed a $(1 - \alpha)\times 100\%$ confidence interval $[\tilde{\cR}_{m}^{\diamond}, \tilde{\cR}_{m}^{\sharp}]$ for the population rank $r_{m}$, where  
\begin{align*}
    \tilde{\cR}_{m}^{\diamond} &= 1 + \sum_{k \neq m} \mathbb{I}\left\{\hat{\theta}_{k} - \hat{\theta}_{m} > \frac{z_{1 - \alpha/2}}{\rho_{m}(\hat{\theta})} + \frac{(1 + c_{0})\sqrt{2\log n}}{\rho_{k}(\hat{\theta})}\right\},\cr
    \tilde{\cR}_{m}^{\sharp} &= n - \sum_{k \neq m} \mathbb{I}\left\{\hat{\theta}_{m} - \hat{\theta}_{k} > \frac{z_{1 - \alpha/2}}{\rho_{m}(\hat{\theta})} + \frac{(1 + c_{0})\sqrt{2\log n}}{\rho_{k}(\hat{\theta})}\right\}. 
\end{align*}
Here $z_{1 - \alpha/2}$ is the $(1-\alpha)$-th quantile of the standard normal distribution, $c_{0} > 0$ is an arbitrary small and fixed positive constant. Note that the uncertainty of $\hat{\theta}_k$ for all $k\ne m$ is controlled by its high-probability Bonferroni bound, leading to the $\sqrt{2\log n}$ multiplier (the worst lower or upper bound). Denote $\tilde{\eta}_{mk} = \frac{z_{1 - \alpha/2}}{\rho_{m}(\hat{\theta})(1 + c_{0})\sqrt{2\log n}} + \frac{1}{\rho_{k}(\hat{\theta})}$ for each $1\leq k\neq m\leq n$. Then it is straightforward that 
\begin{align}\label{eq_Gao_confidence_interval_probability}
    \PP\left(r_{m} \in [\tilde{\cR}_{m}^{\diamond}, \tilde{\cR}_{m}^{\sharp}]\right) \geq \PP\left(\max_{k \neq m} \left|\frac{\hat{\theta}_{k} - \hat{\theta}_{m}}{\tilde{\eta}_{mk}}\right| > (1 + c_{0})\sqrt{2\log n}\right) \geq 1 - \alpha - o(1).
\end{align}
The length of the confidence interval is given by
\begin{align}\label{eq_length_confidence_interval_Gao}
    \tilde{\cL}_{1 - \alpha} = \tilde{\cR}_{m}^{\sharp} - \tilde{\cR}_{m}^{\diamond} = n - 1 - \sum_{k \neq m} \mathbb{I}\left\{\left|\frac{\hat{\theta}_{k} - \hat{\theta}_{m}}{\tilde{\eta}_{mk}}\right| > (1 + c_{0})\sqrt{2\log n}\right\}. 
\end{align}
In contrast, following \eqref{statistics}--\eqref{eq_r_m_construction}, if we set the normalization parameter $\eta_{mk}$ to be $\tilde{\eta}_{mk}$,  our confidence interval for $r_{m}$ is given by $[\cR_{m}^{\diamond}(\tilde{\eta}, \cG_{1 - \alpha}^{\diamondsuit}), \cR_{m}^{\sharp}(\tilde{\eta}, \cG_{1 - \alpha}^{\diamondsuit})]$ with length
\begin{align}\label{eq_length_confidence_interval}
    \cL_{1 - \alpha} = \cR_{m}^{\sharp}(\tilde{\eta}, \cG_{1 - \alpha}^{\diamondsuit}) - \cR_{m}^{\diamond}(\tilde{\eta}, \cG_{1 - \alpha}^{\diamondsuit}) = n - 1 - \sum_{k \neq m} \mathbb{I}\left\{\left|\frac{\hat{\theta}_{k} - \hat{\theta}_{m}}{\tilde{\eta}_{mk}}\right| > \cG_{1 - \alpha}^{\diamondsuit}\right\},
\end{align}
where $\cG_{1 - \alpha}^{\diamondsuit}$ is the $(1 - \alpha)$-th quantile of $\cG^{\diamondsuit} = \max_{k \neq m} \frac{1}{L}\left|\sum_{\ell = 1}^{L}\left(\frac{\hat{\xi}_{k\ell} - \hat{\xi}_{m\ell}}{\tilde{\eta}_{mk}}\right)\omega_{\ell}\right|$. Following the same proof of Theorem~\ref{Theorem_Bootstrap_consistency_simultaneous_inference}, we can similarly obtain
\begin{align*}
    \PP\left(\max_{k \neq m} \left|\frac{\hat{\theta}_{k} - \hat{\theta}_{m}}{\tilde{\eta}_{mk}}\right| > \cG_{1 - \alpha}^{\diamondsuit}\right) \to 1 - \alpha. 
\end{align*}
However, it is easy to verify that the length of our proposed confidence interval in~\eqref{eq_length_confidence_interval} is shorter than that given by \cite{gao2021uncertainty} in~\eqref{eq_length_confidence_interval_Gao}. In detail,  
\begin{align*}
    \PP\left(\tilde{\cL}_{1 - \alpha} \geq \cL_{1 - \alpha}\right) \geq \PP\left(\cG_{1 - \alpha}^{\diamondsuit} \leq (1 + c_{0})\sqrt{2\log n}\right) \to 1. 
\end{align*}
\end{remark}

In conclusion, for any $m \in [n]$, our confidence interval of $r_{m}$ will be narrower than that of~\citet{gao2021uncertainty} with probability tending to $1$. In addition, from the above remark, the confidence interval proposed by \cite{gao2021uncertainty} using high-probability Bonferroni adjustment can also be easily extended to the situation of $M > 2$, once we plug-in the updated formula for $\rho_m(\hat\theta)$ and $\rho_k(\hat\theta)$ as in \eqref{eq_rho_theta} in the normalization parameter $\tilde \eta_{mk}$. However, their Bonferroni confidence interval cannot be directly extended to simultaneous inference of a few ranks whereas in our proposal we have the freedom to choose any interested set $\mathcal M$ of size more than one. 
In our simulation and real data analyses below, we will demonstrate empirically that there is no big difference on  using $\tilde{\eta}_{mk}$ or $\hat\sigma_{mk}$ as the normalization parameter. Thus, we will see the length of confidence intervals for both choices are strictly smaller than the Bonferroni confidence interval constructed in \cite{gao2021uncertainty}.

\subsection{One-Sided Confidence intervals}
\label{Section_Application_Rank_inference}
In this subsection, we provide details on constructing (simultaneous) one-sided intervals for population ranks, and utilizing the one-sided intervals to resolve two more important questions, namely top-$K$ placement testing and sure screening of top-$K$ candidates. This further illustrates the wide applicability of our methodology.

For one-sided intervals, the overall procedure is similar to constructing two-sided confidence intervals. Specifically, let
\begin{align}\label{one_side_statistic}
    \cG^{\circ} = \max_{m \in \mathcal{M}} \max_{k\neq m} \frac{1}{\sqrt{L}} \sum_{\ell = 1}^{L} \left(\frac{\hat{\xi}_{m\ell} - \hat{\xi}_{k\ell}}{\hat{\sigma}_{mk}}\right)\omega_{\ell}
\end{align}
where  $\omega_{1}, \ldots, \omega_{L} \in \RR$ are as before i.i.d.~standard normal random variables. Correspondingly, let $\cG_{1 - \alpha}^{\circ}$ be its $(1 - \alpha)$-th quantile. Under the conditions of Theorem~\ref{Theorem_Bootstrap_consistency_simultaneous_inference}, it follows that   
\begin{align*}
    \left|\mathbb{P}\left\{\max_{m \in \mathcal{M}} \max_{k \neq m} \frac{\sqrt{L}\{\hat{\theta}_{k} - \hat{\theta}_{m} - (\theta_{k}^{*} - \theta_{m}^{*})\}}{\hat{\sigma}_{mk}} > \mathcal{G}_{1 - \alpha}^{\circ}\right\} - \alpha\right| \to 0. 
\end{align*}
Then the $(1 - \alpha)\times 100\%$ simultaneous left-sided confidence intervals for $\{r_{m}\}_{m \in \cM}$ are given by 
\begin{align}\label{eq_one_sided_confidence_interval}
    [\cR_{m}^{\diamond}, n] = \left[1 + \sum_{k \neq m} \mathbb{I}\left\{\hat{\theta}_{k} - \hat{\theta}_{m} > \frac{\hat{\sigma}_{mk}}{\sqrt{L}}\times \cG_{1 - \alpha}^{\circ}\right\}, n\right], \enspace m \in \cM.
\end{align}
We next show the usefulness of the one-sided confidence intervals with two examples. The first one is on testing whether an item of interest lies in the top-$K$ placement. The mathematical framework is given in Example \ref{example1}. 

\begin{example}[Testing top-$K$ placement]\label{example1}
Let $\cM = \{m\}$ for some $1\leq m\leq n$ and let $K \geq 1$ be an prescribed positive integer. We are interested in testing whether $m$-th item is among the top-$K$ ranked items. So the hypotheses under consideration is
\begin{align}\label{eq_top_K_test}
    H_{0} : r_{m} \leq K \enspace \mathrm{versus} \enspace H_{1} : r_{m} > K. 
\end{align}
Based on the one-sided confidence interval $[\cR_{m}^{\diamond}, n]$ in~\eqref{eq_one_sided_confidence_interval}, for any $\alpha \in (0, 1)$, a level~$\alpha$ test for~\eqref{eq_top_K_test} is simply given by 
\begin{align*}
    \psi_{m, K} = \mathbb{I}\{\cR_{m}^{\diamond} > K\}. 
\end{align*}
\end{example}

\begin{proposition}
\label{Proposition_empirical_size}
Under the conditions of Theorem~\ref{Theorem_Bootstrap_consistency_simultaneous_inference}, we have $\PP(\psi_{m, K} = 1|H_{0}) \leq \alpha$ + o(1). In addition, $\PP(\psi_{m, K} = 1|H_{1}) \to 1$ holds when $\theta_{(K)}^{*} - \theta_{m}^{*} \gtrsim \sqrt{\frac{\log n}{{n - 1 \choose M - 1}pL}}$, where $\theta_{(K)}^*$ denotes the underlying score of the item with true rank $K$. 
\end{proposition}

Proposition \ref{Proposition_empirical_size} summarizes the size and power of the above test $\psi_{m, K}$. From Proposition \ref{Proposition_empirical_size}, we are able to control the type-I error effectively below $\alpha$ under the null hypothesis. Moreover, when the alternative holds, the power of the test goes rapidly to one as long as the score difference is larger than the threshold of order $\cO(\sqrt{{\log n}/({n - 1 \choose M - 1}pL}))$. 

Besides testing for the top-$K$ placement, our second crucial example is on constructing a screened candidate set that contains the top-$K$ items with high probability. This is particularly useful in college candidate admission or company hiring decision. Oftentimes, a university or a company would like to design certain admission or hiring policy with the high-probability guarantee of the sure screening of true top-$K$ candidates. We formulate this rigorously in Example \ref{exp:candidate_ad}.
\begin{example}[Candidate admission and sure screening]\label{exp:candidate_ad}
Let $r : [n] \to [n]$ denote the rank operator on $n$ items which maps each item to its population rank and $\cK = \{r^{-1}(1), \ldots, r^{-1}(K)\}$ be the top-$K$ ranked items. Our goal is to select a set of candidates $\hat{\cI}_{K}$ which contains the top-$K$ candidates with a prescribed probability, that is,
\begin{align}\label{eq_cond_IK}
    \mathbb{P}\left( \cK \subset \hat{\mathcal{I}}_{K}\right) \geq 1 - \alpha,
\end{align}
for some $\alpha \in (0, 1)$.  Let $\cM = [n]$ and $\{[\cR_{m}^{\diamond}, n]\}_{m \in [n]}$ denote the corresponding $(1 - \alpha) \times 100\%$ simultaneous left-sided confidence intervals in~\eqref{eq_one_sided_confidence_interval}. Notice that $\cR_{m}^{\diamond} > K$ implies that $r_{m} > K$ for each $m \in [n]$. Hence a natural choice of $\hat{\cI}_{K}$ which satisfies~\eqref{eq_cond_IK} would be 
\begin{align*}
    \hat{\cI}_{K} = \{1\leq m\leq n : \cR_{m}^{\diamond} \leq K\}. 
\end{align*}
In practice, the candidates are more likely to be admitted based on their empirical ranks $\{\hat{r}(1), \ldots, \hat{r}(n)\}$. In this scenario, we then seek the minimal number $D \geq 1$ such that 
\begin{align*}
    \PP\left(\cK \subset \{\hat{r}^{-1}(1), \ldots, \hat{r}^{-1}(D)\}\right) \geq 1 - \alpha. 
\end{align*}
To estimate $D$, with a slight abuse of notation, set $\hat{\sigma}_{mk}= 1$ for all $1\leq k\neq m\leq n$ in \eqref{one_side_statistic} and denote the corresponding $(1 - \alpha)\times 100\%$ simultaneous left-sided confidence intervals as $\{[\hat{\cR}_{m}^{\diamond}, n]\}$. Then, similar to $\hat{\cI}_{K}$, our estimator for $D$ is defined by  
\begin{align*}
    \hat{D} = \max \left\{1\leq m\leq n : \hat{\cR}_{m}^{\diamond} \le K\right\},
\end{align*}
where $\hat{\cR}_{1}^{\diamond} \leq \hat{\cR}_{2}^{\diamond} \leq \ldots \leq \hat{\cR}_{n}^{\diamond}$ maintains the same ranking as the empirical ranks.
\end{example}

In Example \ref{exp:candidate_ad}, we constructed sure screening set for top-$K$ candidates. Although many works studied sure screening property of the high-dimensional regression coefficients \citep{fan2008sure,fan2010sure,zhu2011model,li2012robust,barut2016conditional,wang2016high,fan2022latent}, the study on the sure screening property of population ranks is much less explored and to our best knowledge, our procedure is the first one to have concrete theoretical guarantee.

\begin{figure}[t]
    \centering
    \begin{tabular}{cc}
   \hskip-30pt \includegraphics[width=0.45\textwidth]{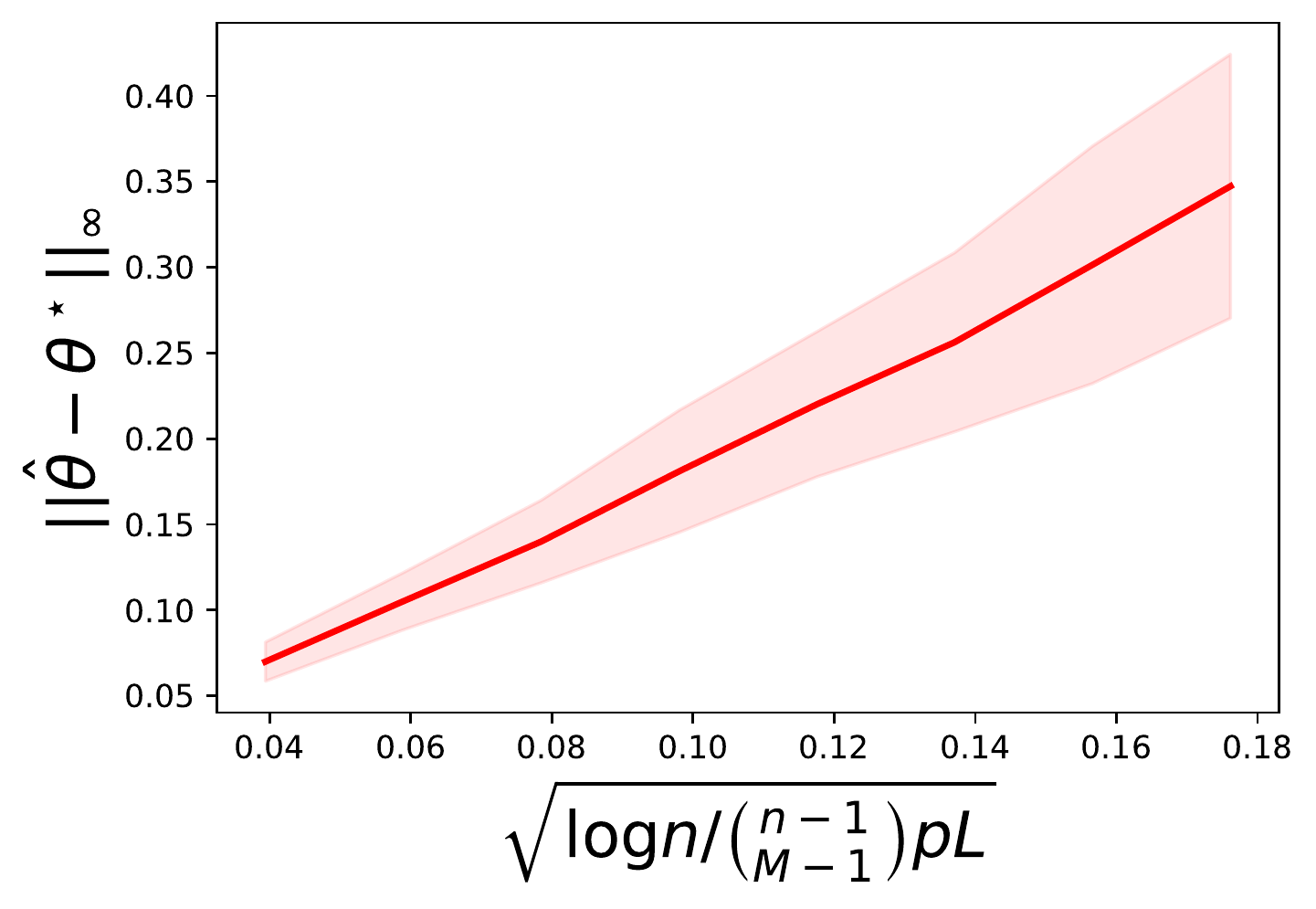}
   &\hskip-5pt \includegraphics[width=0.45\textwidth]{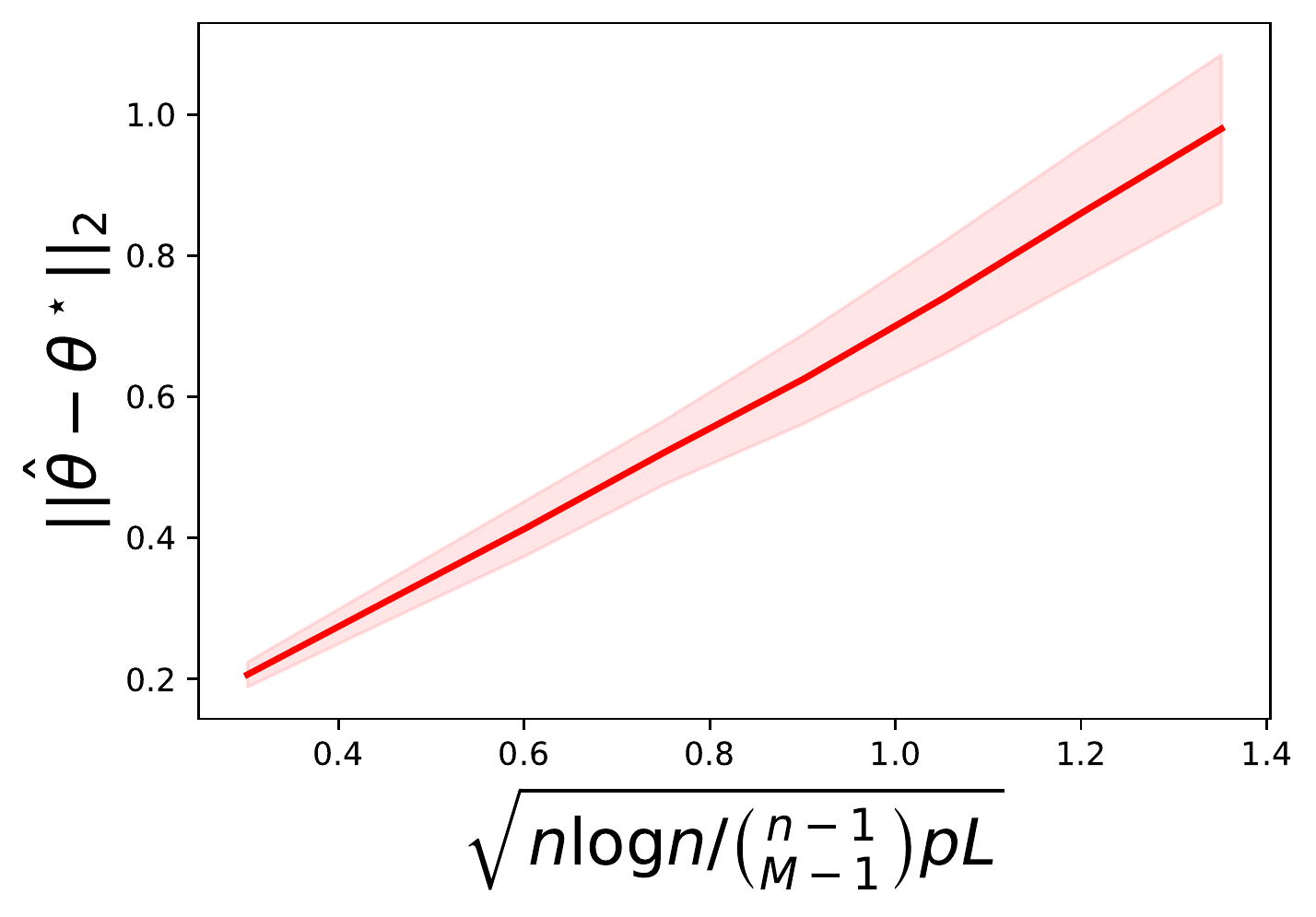}\\
   (a)&(b)
   \end{tabular}
    \caption{$\ell_{\infty}$- and $\ell_2$- statistical errors of the MLE $\hat\theta$ against the theoretical rate when $p$ varies. The solid lines represent the averaged statistical errors of 200 repetitions and the light areas denote the standard deviations.}
    \label{fig:consist_p}
\end{figure}

\begin{figure}[t]
    \centering
    \begin{tabular}{cc}
   \hskip-30pt \includegraphics[width=0.45\textwidth]{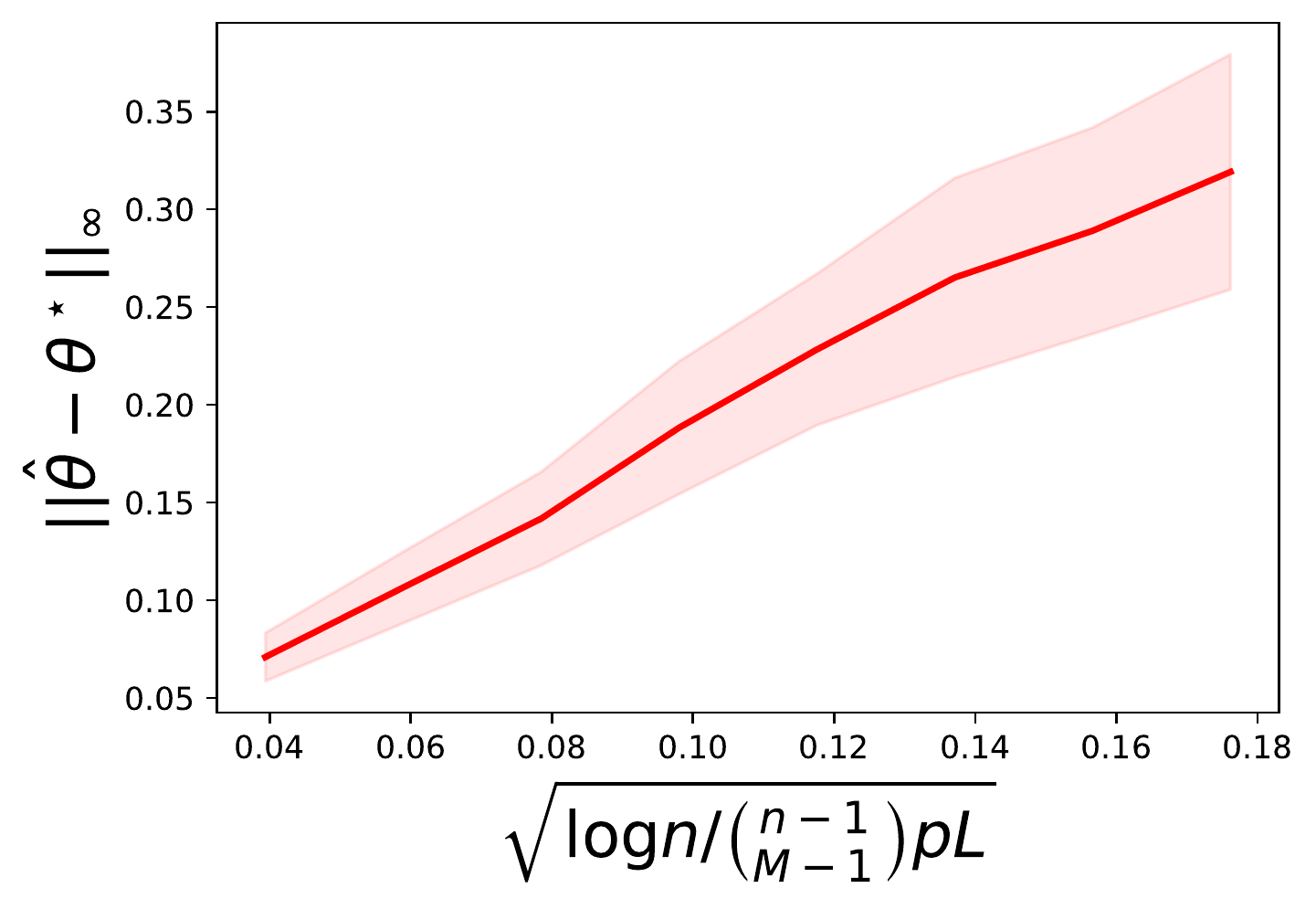}
   &\hskip-5pt \includegraphics[width=0.45\textwidth]{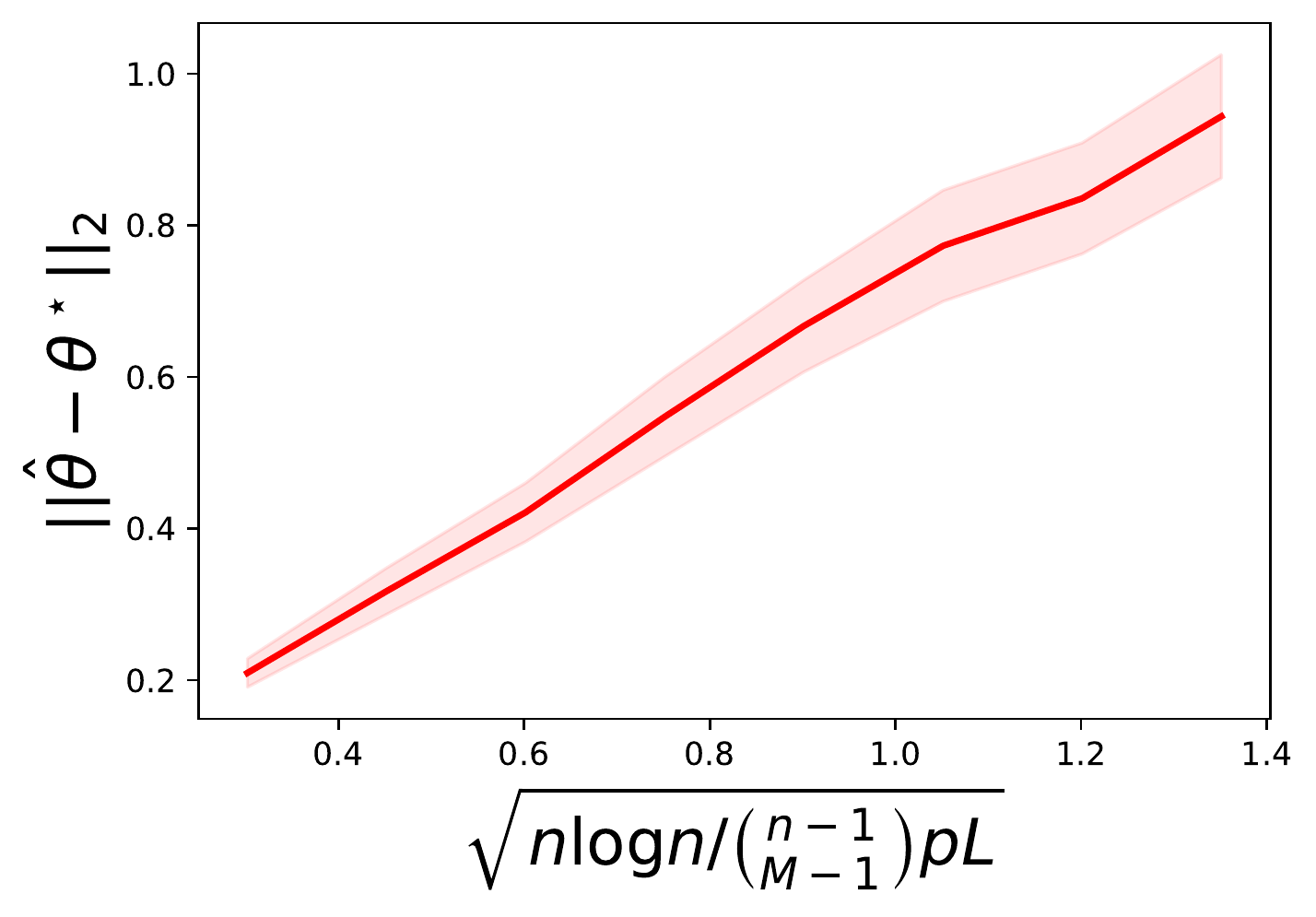}\\
   (a)&(b)
   \end{tabular}
    \caption{$\ell_{\infty}$- and $\ell_2$- statistical errors of the MLE $\hat\theta$ against the theoretical rate when $L$ varies. The remaining captions are the same as those in Figure \ref{fig:consist_p}.}
    \label{fig:consist_l}
\end{figure}

\section{Numerical Studies} \label{sec:numerical}

In this section, we first conduct numerical studies via synthetic data to demonstrate our theoretical results in finite samples in Section~\ref{sec:numerical5.1} and to illustrate the effectiveness of our proposed confidence interval methodology in Section~\ref{sec:numerical5.2}. We then analyze a real data example in Section~\ref{sec:numeric5.3}.

\subsection{Synthetic Data Analysis}\label{sec:numerical5.1}
To confrim the correctness of our theorem in finite samples, we will verify the statistical convergence rate, the asympototic normality and the effectiveness of Gaussian multiplier bootstrap in this subsection. 

\medskip
\textbf{Statistial rates of convergence.} 
We first validate the statistical rates of our MLE estimator $\hat\theta$ in both $\ell_{\infty}$- and $\ell_{2}$-norms.
In the first simulation, we fix $n=60,M=3,L=20$ and let $p$ vary such that $\sqrt{\log n/\binom{n-1}{M-1}pL}$ takes uniform grids from $0.04$ to $0.18$. Meanwhile, we generate every entry of the true $\theta^*$ independently from Uniform$[2,4]$. We then record the $\ell_{\infty}$- and $\ell_2$- statistical errors of $\hat\theta$ to $\theta^*$ by solving the MLE given in \eqref{likelihood_main}. The average errors together with the standard deviations of $200$ repetitions for each $p$ are displayed in Figure \ref{fig:consist_p}. 
In the second simulation, we investigate the effects of $L$ on these statistical rates. In this scenario, we fix $n=60,M=3,p=0.05$ and let $L$ vary such that $\sqrt{\log n/\binom{n-1}{M-1}pL}$ takes uniform grids from $0.04$ to $0.18$. The remaining procedures are the same as above and the results are shown in Figure \ref{fig:consist_l}.
Clearly, we observe from Figures \ref{fig:consist_p} and \ref{fig:consist_l}, the statistical rates are proportional to their theoretical rates, as indicated by the overall linear pattern. These simulation results lend further support to the theoretical results in Theorem \ref{thm_consistency}.

\begin{figure}[]
    \centering
    \begin{tabular}{c}
    \includegraphics[width=\textwidth]{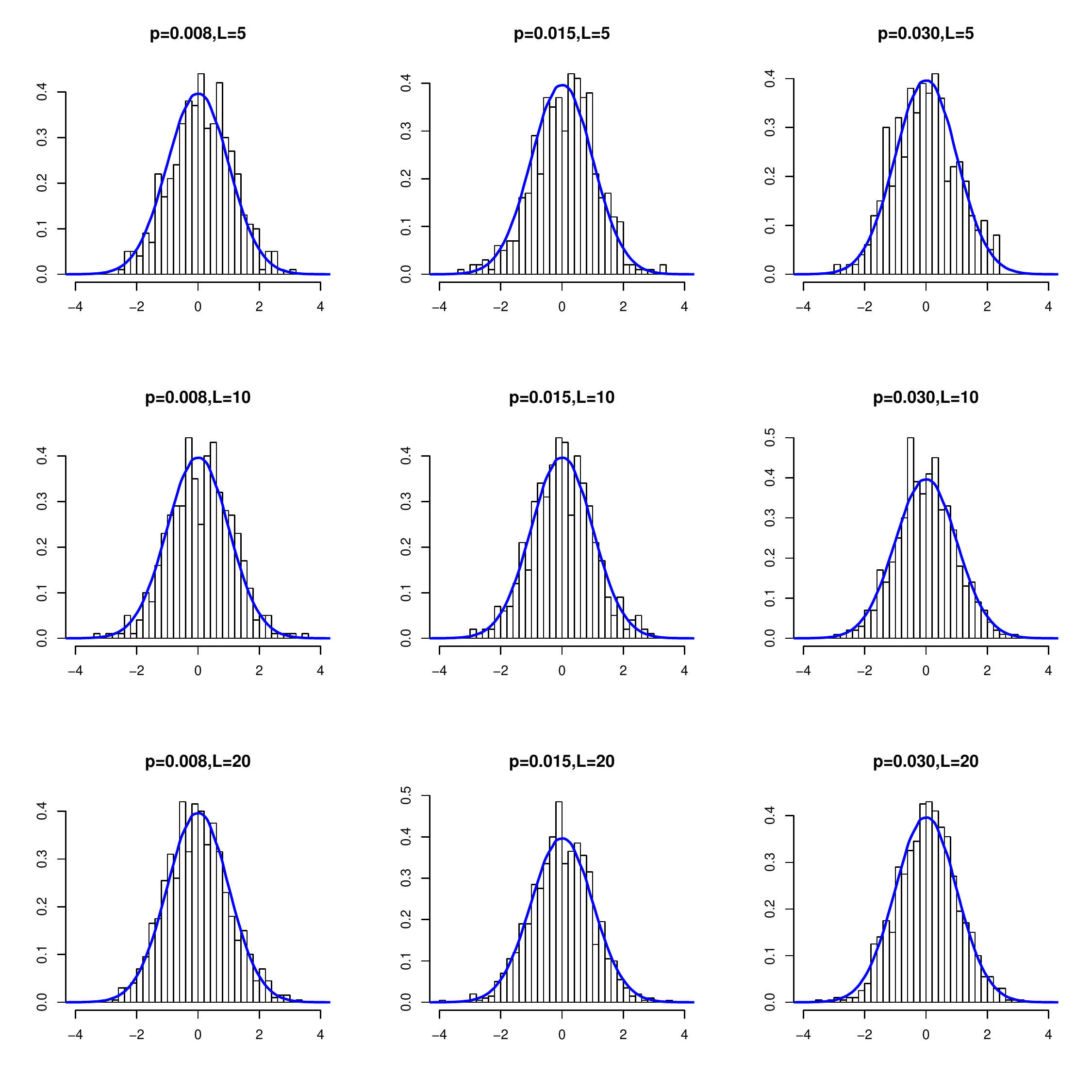}
   \end{tabular}
    \caption{Histogram of $\rho_1(\hat\theta)(\hat\theta_1 - \theta_1^*)$ against standard Gaussian distribution. The blue curve denotes the standard Gaussian distribution. The experiments were repeated for 500 times for each combination of $L\in \{5,10,20\}$ and $p \in \{0.008, 0.015, 0.03\}$.}
    \label{fig:qqplot}
\end{figure}

\medskip
\textbf{Asymptotic normality.}
Next, we investigate the uncertainty quantification of the MLE estimator $\hat\theta$. Here we fix $n=60,M=3$ and choose $L,p$ from $\{5,10,20\}$ and $\{0.008,0.015,0.03\}$ respectively, which results in $9$ combinations.  For each combination, the true $\theta^*$ is generated independently from Uniform$[2,4]$ for $500$ times. We record the empirical distributions of the standardized $\hat\theta_1$, that is $\rho_1(\hat\theta)(\hat\theta_1 - \theta_1^*)$, of these $500$ repetitions, and check its normality via histograms, presented in Figure \ref{fig:qqplot}.  
From Figure \ref{fig:qqplot}, we observe that the empirical distribution of $\hat\theta_1$ is well approximated by the standard Gaussian distribution, especially when we have a larger $p$ and $L.$ This is consistent with the theoretical results in Theorem \ref{thm_inference_k3}.
 
\begin{figure}[t]
    \centering
    \begin{tabular}{c}
    \includegraphics[width=0.7\textwidth]{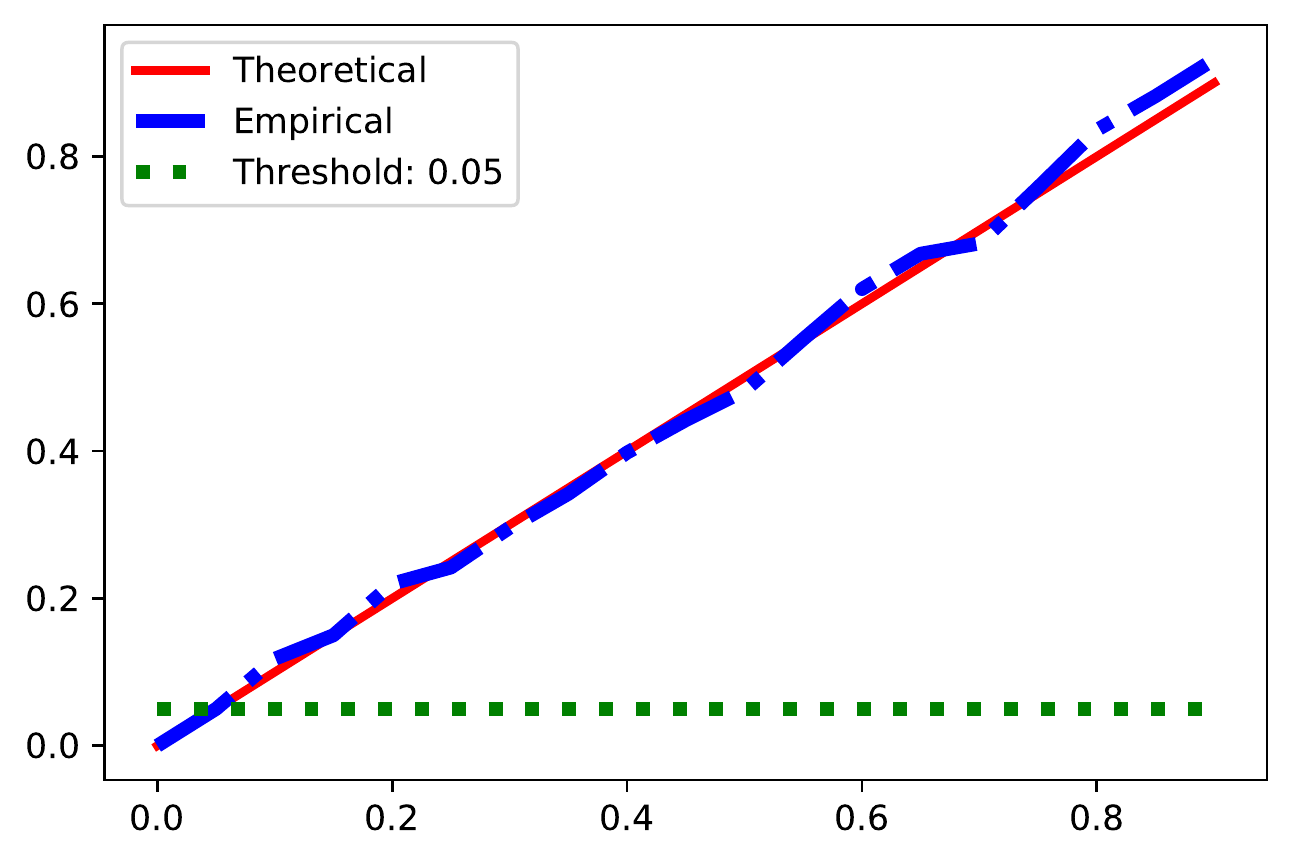}
   \end{tabular}
    \caption{PP-plot of empirical probability $\hat{\mathbb{P}}(\cT > \cG_{1-\alpha})$ of $\cT$ given in \eqref{statistics} with $\cM=\{1\}$ against theoretical significance level $\alpha$. The red solid and blue dash-dotted lines represent theoretical and empirical probabilities respectively. The green dotted line represents the the case when significance is chosen to be $0.05$. }
    \label{fig:pp_plot}
\end{figure}
 
\medskip
\textbf{Gaussian multiplier bootstrap.} Finally, we validate the Gaussian approximation results discussed in Section~\ref{sec:ranking_inference}. We let $n=60,M=3,L=80,p=0.05$ and investigate the distribution of $\cT$ in \eqref{statistics} with $\cM=\{1\}$. By Theorem~\ref{Theorem_Bootstrap_consistency_simultaneous_inference}, we have $\mathbb{E}[\mathbb{I}\{\cT > \cG_{1-\alpha}\}] = \mathbb{P}(\cT > \cG_{1-\alpha}) \to \alpha$. We then verify this with various $\alpha \in \{0.05,0.10,\cdots,0.90\}$. 
For every $\alpha,$ we bootstrap $300$ times to compute the critical value $\cG_{1-\alpha}$, and repeat this whole procedure $500$ times to calculate $\hat{\mathbb{P}}(\cT > \cG_{1-\alpha}) = \frac{1}{500} \sum_{b=1}^{500} \mathbb{I}\{\cT^b > \cG_{1-\alpha}^b\}$ where $\cT^b$ and $\cG_{1-\alpha}^b$ are the pairwise score difference statistic and the corresponding bootstrap critical value of the $b$-th repetition. Note that in each repetition, the true $\theta^*$ is generated independently from Uniform$[2,4]$ as before. 
Note that $\hat{\mathbb{P}}(\cT \le \cG_{1-\alpha})$ is the empirical coverage probability for the pairwise score difference statistic $\cT$. 
Figure \ref{fig:pp_plot} gives the so-called PP-plot, which shows $\hat{\mathbb{P}}(\cT > \cG_{1-\alpha})$ against the theoretical significance level. 
From Figure  \ref{fig:pp_plot}, it the clear that the empirical probability match well with the theoretical ones. Especially, when we apply the significant level of $\alpha=0.05$, the empirical exceptional probability is indeed around $0.05$.

\subsection{Confidence Interval}\label{sec:numerical5.2}
Below we provide numerical studies for validating our proposed framework for confidence interval construction in Section~\ref{sec:ranking_inference}. Throughout this subsection, we use $n=60,M=3,L=80$ and let $p$ vary in $\{0.05,0.10,0.15\}.$  In addition, the entries of $\theta\in \RR^{60}$ are generated from the uniform grids from $4$ to $2$, i.e. $\theta_i = 4 - (i-1)/30$ for $i=1,\dots,60$. For every experiment, we conduct the bootstrap for 500 times to compute the critical value according to $\alpha=0.05$ and further repeat the entire procedure for 500 times.

\medskip
\textbf{Two-sided confidence intervals.} We construct the confidence interval (CI) for the $10$-th item ($\mathcal{M}=\{10\}$) using our method in Section~\ref{sec:ranking_inference} and the Bonferroni correction proposed in \cite{gao2021uncertainty} respectively. Concretely, we will compare the following 3 confidence intervals: (i) our bootstrap CI $[\cR_{10}^{\diamond}(\hat\sigma, \cG_{1 - \alpha}), \cR_{10}^{\sharp}(\hat\sigma, \cG_{1 - \alpha})]$ with the normalization $\hat\sigma_{mk}$, (ii) our bootstrap CI $[\cR_{10}^{\diamond}(\tilde{\eta}, \cG_{1 - \alpha}^{\diamondsuit}), \cR_{10}^{\sharp}(\tilde{\eta}, \cG_{1 - \alpha}^{\diamondsuit})]$ but with the normalization $\tilde\eta_{mk}$ ($c_0 = 1$), (iii) the Bonferroni CI $[\tilde{\cR}_{10}^{\diamond}, \tilde{\cR}_{10}^{\sharp}]$ given by \cite{gao2021uncertainty} extended to $M=3$ in this simulation (see discussions after Remark~\ref{remark4.1}). For each choice of the above confidence intervals, we report: (a) EC($\theta$) -- the empirical coverage probability $\hat{\mathbb{P}}(\cT \le \cG_{1-\alpha})$ for $\cT$, which is also the overall empirical coverage probability for all the score differences $\theta_{k}^* - \theta_{10}^*, \forall k \ne 10$, (b) EC($r$) -- the empirical coverage of the confidence interval constructed for rank $r_{10}$, and furthermore (c) Length -- the length of the CI for rank $r_{10}$ which equals $\cR_{10}^{\sharp} - \cR_{10}^{\diamond}$, where $[\cR_{10}^{\diamond}, \cR_{10}^{\sharp}]$ is any of the above CI in (i)-(iii).
  
    \begin{table}[t]
	\begin{center}
		\def\arraystretch{1.2}
		\setlength\tabcolsep{4pt}
		\begin{tabular}{c||c|c|c||c|c|c||c|c|c}
			\toprule
			CI & \multicolumn{3}{c||}{$[\cR_{10}^{\diamond}(\hat\sigma, \cG_{1 - \alpha}), \cR_{10}^{\sharp}(\hat\sigma, \cG_{1 - \alpha})]$} & \multicolumn{3}{c||}{$[\cR_{10}^{\diamond}(\tilde{\eta}, \cG_{1 - \alpha}^{\diamondsuit}), \cR_{10}^{\sharp}(\tilde{\eta}, \cG_{1 - \alpha}^{\diamondsuit})]$} & \multicolumn{3}{c}{$[\tilde{\cR}_{10}^{\diamond}, \tilde{\cR}_{10}^{\sharp}]$} \\
            \hline
			  & EC($\theta$) & EC($r$) & Length  & EC($\theta$) & EC($r$) & Length & EC($\theta$) & EC($r$) & Length \\
			\hline
			$p=0.05$	&0.950
			 &1.000 &5.590 &0.948
			 &1.000 &5.716  &1.000 &1.000 &10.290 \\
			$p=0.10$	&0.940 &1.000 & 3.604
			 &0.956 &1.000  & 3.686
			 &1.000 &1.000 &6.962 \\
			$p=0.15$	  &0.950 &1.000  &2.886   &0.948 &1.000 &2.928  &1.000 &1.000 &5.476 \\
			\bottomrule
		\end{tabular}
	\end{center}
	\caption{Empirical coverages and lengths of two-sided confidence intervals (CI). For any choice of CI, ``EC($\theta$)'' and ``EC($r$)'' denote the empirical coverages of confidence intervals for score differences and ranks. ``Length'' denotes the length of the CI for $r_{10}$. The first two CI's use our bootstrap method in Section~\ref{sec:ranking_inference} and only differ in the normalization parameters used. The third CI is based on the Bonferroni method of \cite{gao2021uncertainty}. 
	All numbers are averaged over 500 replications.}
	\label{tab1}
\end{table}

The results of the empirical coverages for score differences and ranks and the length of confidence intervals are summarized in Table \ref{tab1}. Table \ref{tab1} reveals that the empirical coverage probability for score differences, no matter which normalization it uses, is approximately $0.95$, which is consistent with our theory in Section~\ref{sec:ranking_inference}. 
However, the CI for the rank of the $10$-th item is more conservative since the rank must be an integer, so we see the empirical coverage of the confidence intervals for the rank stays at one for all cases. In comparison, the confidence interval via \cite{gao2021uncertainty}'s method is even more conservative as the empirical coverages of the confidence intervals for $\theta$ are already one and the lengths of their confidence intervals are much wider than ours.

\begin{table}[t]
	\begin{center}
		\def\arraystretch{1.2}
		\setlength\tabcolsep{4pt}
		\begin{tabular}{c||c|c|c||c|c|c|c|c}
			\toprule
			 & \multicolumn{3}{c||}{Null holds: $r_m \le K$} & \multicolumn{5}{c}{Alternative holds: $r_m > K$} \\
			 \hline
			$m$ & $K-2$ & $K-1$ & $K$ & $K+1$ & $K+2$ & $K+3$ &$K+4$ & $K+5$\\
			$\theta_{m}^*-\theta_{K}^*$ & $2/30$ & $1/30$ & $0$ & $-1/30$ & $-2/30$ & $-3/30$ &$-4/30$ & $-5/30$\\
			\hline 
			$p=0.05$	&0 (0.036)
			 &0 (0.038) &0 (0.050)  &0.008 &0.144 &0.444 &0.822 &0.986\\
			$p=0.10$	&0 (0.043) &0 (0.045)  & 0 (0.054)
			 &0.032 &0.372 &0.896 &0.993 &1\\
			$p=0.15$	  &0 (0.042)  &0 (0.038) &0 (0.046) &0.094 &0.624 &0.984 &1 &1\\
			\bottomrule
		\end{tabular}
	\end{center}
	\caption{Sizes and powers of the test $\psi_{m, K} = \mathbb{I}\{\cR_{m}^{\diamond} > K\}$ for testing hypothesis in \eqref{eq_top_K_test}. The numbers inside the bracket indicate the sizes of testing score differences as we discussed in the text. The numbers outside the bracket represent the sizes or powers of directing testing the rank, which is our goal here.  All displayed numbers are averaged over 500 replications. }
	\label{tab2}
\end{table}

\medskip
\textbf{One-sided confidence intervals.} Next we validate the testing performance of the test $\psi_{m, K} = \mathbb{I}\{\cR_{m}^{\diamond} > K\}$ for \eqref{eq_top_K_test}. 
In this experiment, we choose $K=10$ in \eqref{eq_top_K_test} and by default we would like to use the normalization parameter $\hat\sigma_{mk}$ in constructing $\cR_{m}^{\diamond}$ as in \eqref{eq_one_sided_confidence_interval}. Consider $m \in \{K-2,\cdots,K+5\}$.  We computed the proportion of rejection for each given $m$.   If $m \le K$ the null hypothesis is true and the proportion is approximately the sizes of the test, whereas if the alternative is true, the proportion is approximately the power of the test. In addition, when the null holds, we also calculated the size of the test $\mathbb{I}\{\max_{k \neq m} {\sqrt{L}\{\hat{\theta}_{k} - \hat{\theta}_{m} - (\theta_{k}^{*} - \theta_{m}^{*})\}}/{\hat{\sigma}_{mk}} > \mathcal{G}_{1 - \alpha}^{\circ}\}$ for the one-sided hypotheis testing problem $H_0: \theta_{m}^* - \theta_k^* \ge (k-m)/30$ for all $k \ne m$ for testing score differences.

The results are presented in Table~\ref{tab2}. We observe from this table that when the alternative holds and true rank increases, the power of our test rise rapidly to $1$. On the other hand, when the null hypothesis holds, we control the test size around $\alpha=0.05$ if we are testing the score differences and the rank test becomes more conservative and get size equal to zero, which can be a good feature in practice. 

\begin{table}[t]
	\begin{center}
		\def\arraystretch{1.2}
		\setlength\tabcolsep{4pt}
		\begin{tabular}{c|c|c|c|c|c}
		\toprule
		 &EC($\theta$) & EC($r$) & $K=5$ & $K=10$ & $K=15$  \\
			\hline 
			 			$p=0.05$ &0.964 &1.000	
			 &8.81  &13.82  &18.89  \\
			$p=0.10$	&0.966 &1.000
			 &7.55  &12.49  &17.54 \\
			$p=0.15$	 &0.962 &1.000 &6.98  
			 &12.00  &17.00 \\
			\bottomrule
		\end{tabular}
	\end{center}
	\caption{Empirical coverages and lengths of sure screening confidence set for Example \ref{exp:candidate_ad}. ``EC($\theta$)'' represents the empirical coverage of confidence intervals for score differences. Note EC($\theta$) does not depend on $K$. ``EC($r$)'' denotes the empirical coverage of the confidence intervals for true top $K$ items. For all $K\in \{5,10,15\}$, we see $\textrm{EC}(r) = 1$, so we collapse them into just one column. Other numbers are the lengths of the sure screening confidence set $\hat{\cI}_{K}$.
	All displayed numbers are averaged over 500 replications. }
	\label{tab3}
\end{table}

\medskip
\textbf{Uniform one-sided confidence intervals.} In the candidate admission Example~\ref{exp:candidate_ad}, to guarantee the sure screening property, we need to build uniform one-sided coverage for all items, i.e. $\mathcal{M} = [n]$. Following Example~\ref{exp:candidate_ad}, once we have the one-sided confidence interval for all items, the sure screening set is given by $\hat{\cI}_{K} = \{1\leq m\leq n : \cR_{m}^{\diamond} \leq K\}$. In this simulation, we choose $K=\{5,10,15\}$. 
In Table~\ref{tab3}, we report the average length of $\hat{\cI}_{K}$ and empirical coverages of the confidence intervals for score differences and for true top-K ranks over 500 replications. It turns out that in the simulation the length of the sure screening confidence set is no more than $K+4$ even when the sampling probability $p$ is as small as $0.05$ since the true scores are well-separated. 
Similar to what we have previously seen, the rank empirical coverage probabilities ``EC($r$)'' are again all one for $K = 5,10,15$, indicating that the sure screening set is already conservative.
Finally, since we use $\mathcal{M} = [n]$, the empirical coverage ``EC($\theta$)'' for score differences is independent of $K$. From Table~\ref{tab3}, for different $p$, this coverage via the Gaussian multiplier bootstrap is approximately $0.95$ as expected.

\subsection{Real Data Analysis}\label{sec:numeric5.3}
In this subsection, we analyze a real dataset to corroborate the practical effectiveness of our proposed methodology and its associated theoretical guarantees. 
We choose to use the relatively simple \emph{Jester Dataset} \citep{goldberg2001eigentaste} which contains ratings for 100 jokes from $73,421$ users and is available on the website of 
https://goldberg.berkeley.edu/jester-data/. Among all the users, $14,116$ rated all $100$ jokes. Our analyses are based upon these users who rated all jokes for simplicity. 

\medskip
\textbf{Results with $M = 3$.} As for data generation, we first synthesize an uniform-hypergraph with edge sampling probability $p=0.05$ and consider the setting $M=3$, namely, any 3 jokes are chosen for comparisons with probability $p=0.05$. For any selected tuple, we randomly select $L=80$ rankings from those $14,116$ users who ranked all jokes, and observe the top rankings for the selected tuple. For CI construction, we follow our methodology discussed in Section~\ref{sec:ranking_inference}. We construct both two-sided and one-sided $1-\alpha$ confidence intervals for top-15 ranked items with $\alpha=0.05$. We summarize our inference results in Table \ref{rank_interval}.

\begin{table}[ht]
	\begin{center}
		\begin{tabular}[t]{ p{1.5cm}|p{1.5cm} p{1.5cm}| p{1.5cm}p{1.5cm}p{1.5cm}|p{1.5cm}p{1.5cm} }
			\toprule
			Rank &Joke ID & Score & $TC_1$ & $TC_2$ & $TC_3$ & $OC$ &$UOC$ \\ \hline
			1 &89 &0.89 &[1,2] &[1,2] &[1,2] &[1,100] &[1,100]  \\ 
			2  &50 &0.85 &[1,2] &[1,2] &[1,2] &[1,100] &[1,100]  \\ 
			3  &27 &0.73 &[3,8] &[3,8] &[3,9] &[3,100] &[3,100]  \\ 
			4 &36 &0.69 &[3,8] &[3,8] &[3,10] &[3,100]  &[3,100] \\ 
			5  &35 &0.69 &[3,9] &[3,9] &[3,10] &[3,100] &[3,100]  \\ 
			6  &29 &0.67 &[3,9] &[3,9] &[3,11] &[3,100] &[3,100] \\ 
			7  &32 &0.67 &[3,9] &[3,9] &[3,11] &[4,100]  &[3,100] \\ 
			8  &62 &0.66 &[3,10] &[3,10] &[3,11] &[4,100] &[3,100] \\ 
			9  &54 &0.64 &[4,11] &[5,11] &[3,14] &[5,100] &[5,100] \\ 
			10  &53 &0.60 &[8,12] &[8,12] &[4,16] &[9,100] &[7,100]  \\ 
			11  &49 &0.57 &[9,15] &[9,15] &[6,16] &[10,100] &[9,100]  \\ 
			12  &68 &0.54 &[10,16] &[10,16] &[9,21] &[11,100]  &[9,100]\\ 
			13  &72 &0.52 &[11,16] &[11,16] &[9,21] &[11,100]  &[10,100] \\ 
			14  &66 &0.52 &[11,16] &[11,16] &[9,21] &[11,100] &[10,100]  \\ 
			15  &69 &0.51 &[11,16] &[11,16] &[10,21] &[11,100] &[10,100] \\ 
			\bottomrule
		\end{tabular}
		\label{twoquantities}
	\end{center}
	\caption{Confidence intervals for ranks of jokes in the \it{Jester Dataset}. The columns \emph{Rank}, \emph{Score} and \emph{Joke ID} denote the estimated rank, estimated MLE score and the corresponding joke IDs, respectively. $TC_1,TC_2$ denote our bootstrap two-sided confidence intervals constructed using $\hat\sigma_{mk}$  and $\tilde{\eta}_{mk}$ ($c_0=1$) as the normalization parameters. $TC_3$ is constructed based on the Bonferroni method of \cite{gao2021uncertainty}. Furthermore, $OC$ represents the one-sided confidence intervals for each individual item and $UOC$ denotes the uniform one-sided confidence intervals for all items together, which are used to generate the sure screening confidence set in Example~\ref{exp:candidate_ad}. }
	\label{rank_interval}
\end{table}

From Table \ref{rank_interval}, we observe that our bootstrap method using $\hat\sigma_{mk}$ ($TC_1$) or $\tilde{\eta}_{mk}$ ($TC_2$) does not make too much difference. However, it is worth mentioning that our confidence intervals constructed using either of these two normalization parameters are strictly better than the confidence interval constructed via the Bonferroni method ($TC_3$) of \cite{gao2021uncertainty}. Furthermore, we also build the one-sided confidence intervals for each individual item, denotes as $OC$, and the uniform one-sided confidence intervals for all items together, denoted as $UOC$, which is wider than $OC$ due to the overall control. $OC$ can be used to conduct the hypothesis testing in \eqref{eq_top_K_test}. For example, if we care to test whether a joke is within the top-$3$ funniest in this real data, we will reject the hypothesis from the $7$-th ranked item. If we test whether a joke is within the top-$10$ best, we will reject the hypothesis from the $12$-th ranked item. $UOC$ can be used to generate the sure screening confidence set in Example~\ref{exp:candidate_ad}. For example, a set that contains all the top-$5$ jokes with high probability should include the first $9$ jokes in total. 

\begin{table}[t]
	\begin{center}
		\begin{tabular}[t]{ p{1.5cm} | p{1.5cm}p{1.5cm}p{1.5cm}p{1.5cm}p{1.5cm} p{1.5cm}}
			\toprule
			 Joke ID  & $M=2$ & $M=3$ & $M=4$ & $M=5$ &$M=6$ \\ \hline
			 10  &[24,51] &[23,50] &[26,48] &[21,45] &[24,43]  \\ 
			  30 &[63,89] &[78,94] &[80,97] &[79,95] &[84,98]  \\ 
			  50  &[1,4] &[1,4] &[1,4] &[1,4] &[2,4]  \\ 
			 70  &[52,81] &[63,79] &[63,80] &[67,81]  &[76,89] \\ 
			  90  &[47,74] &[40,66] &[43,70] &[41,68] &[38,64]  \\
			\bottomrule
		\end{tabular}
		\label{twoquantities}
	\end{center}
	\caption{Two-sided confidence intervals constructed for jokes with ID in $\{10,30,50,70,90\}$. We choose the sampling probability $p_M$  such that ${M}/[{\binom{n-1}{M-1}p_M}]$ is a fixed constant for $M\in \{2,\cdots 6\}$. }
	\label{tab_M}
\end{table}

\medskip
\textbf{Results with different $M$'s.} According to our asymptotic distribution, the asymptotic variance is of the order ${M}/{[\binom{n-1}{M-1}p_M]}$ when we assume $\{\theta_i\}_{i\in[n]}$ are in the same order. Now we let $M$ vary and choose $p_M$ such that ${M}/{[\binom{n-1}{M-1}p_M]}$ is a fixed number. Specifically, we fix $p_2=0.3$, we will have the following $\{p_M\}_{M\in \{3,\cdots 6\}}$: $p_3=9\times 10^{-3},p_4=3.7\times 10^{-4},p_5=1.9\times 10^{-5},p_6=1.2\times 10^{-6}$. We summarize the two-sided confidence intervals for jokes with ID in $\{10,30,50,70,90\}$ for each $M$ in Table \ref{tab_M}.
We observe from Table \ref{tab_M}, when we increase $M$ but keep the same ${M}/{\binom{n-1}{M-1}p_M}$, we still obtain confidence intervals with comparable length for any given item. We also observe that it allows a much smaller sampling probability $p_M$ when $M$ is large to construct confidence intervals of the same significance level.

\begin{table}[t]
	\begin{center}
		\begin{tabular}[t]{ p{1.5cm} | p{1.5cm}p{1.5cm}|p{1.5cm}p{1.5cm}|p{1.5cm} p{1.5cm}|p{1.5cm} p{1.5cm}}
			\toprule
			& \multicolumn{2}{c|}{$p_{2,3}=1\times 10^{-1}$} & \multicolumn{2}{c|}{$p_{3,4}=3\times 10^{-3}$}& \multicolumn{2}{c|}{$p_{4,5}=1.2\times 10^{-4}$}& \multicolumn{2}{c}{$p_{5,6}=6\times 10^{-5}$}\\
			\hline
			 Joke ID  & $M=2$ & $M=3$ & $M=3$ & $M=4$ &$M=4$ &$M=5$ &$M=5$ &$M=6$ \\ \hline
			 10  &[12,69] &[33,39] &[17,63] &[28,45] &[15,50] &[29,36]  &[15,47] &[27,36]\\ 
			  30 &[54,98] &[83,87] &[66,99] &[86,95] &[78,99] &[88,94] &[78,99] &[88,97]\\ 
			  50  &[1,10] &[1,2] &[1,14] &[2,4] &[1,5]  &[2,2] &[1,7] &[2,2]\\ 
			 70  &[45,92] &[68,75] &[50,90] &[73,81]  &[51,87] &[76, 79]&[61,97] &[77,80]\\ 
			  90  &[32,88] &[55,63] &[40,84] &[55,68] &[35,77] &[51,63] &[31,72] &[52,64] \\
			\bottomrule
		\end{tabular}
		\label{twoquantities}
	\end{center}
	\caption{Two-sided confidence intervals constructed for jokes with ID in $\{10,30,50,70,90\}$. We compare the confidence intervals for every two adjacent $M$ under the same sampling probability. }
	\label{tab_M2}
\end{table}

As noted before, for each given sampling probability $p$, the effective number of samples is very different for different $M$.  Therefore, we compare the inference results  only for  the adjacent $M$ and $M+1$ with the same $p$, denoted as $p_{M, M+1}$.  Specifically, we pre-select 5  jokes and compute their  two-sided confidence intervals based on $M$-way and $(M+1)$-way  comparisons with a fixed sampling probability $p_{M,M+1},M\in \{2,3,4,5\}$.     The results are presented in Table \ref{tab_M2}.
We observe that for a fixed $p_{M,M+1},$ the two-sided confidence intervals with $M+1$ are much narrower than those with $M$. Moreover, for a given $M$, if we increase $p$, the lengths of confidence intervals also become much smaller.   Both of these conclusions are due to the increase of sample size in both scenarios.

\section{Conclusion and Discussion}\label{sec:discussion}
This paper studies the ranking inference problem based on multiway comparisons. Unlike the conventional Plackett-Luce model \citep{plackett1975analysis}, which models the entire multiway rankings, we considered the more general case of only observing the top choices. 
Such a model serves as an extension of the famous Bradley-Terry-Luce model and modifies the Plackett-Luce model in a useful and practical direction. Theoretically, under the sparsest uniform sampling regime, we proposed to estimate the underlying preference scores via the MLE and established its \emph{optimal} $\ell_2$- and $\ell_{\infty}$- statistical rates. This closed the gap of achieving the optimal convergence with a practical algorithm under the sparest comparison hypergraph. 
Moreover, little has been done to quantify the asymptotic uncertainty of an estimator for the multiway comparisons. To our best knowledge, our work is the first to derive and justify the asymptotic distribution of the MLE for the underlying preference scores in the top-choice multiway comparison model. We should emphasize again that our theoretical contributions are highly nontrivial as the justification for general $M$ is quite mathematically involved. More importantly, we proposed a novel inference framework for building confidence intervals for ranks, which are provably narrower than the confidence intervals with high-probability Bonferroni correction in \cite{gao2021uncertainty}. This framework is valuable in solving outstanding inference questions, including testing top-$K$ placement and constructing sure screening confidence sets. 

There are a few future directions to improve our work further. Firstly, we studied the ranking problem based on a uniform comparison hypergraph. That is, each comparison is made among $M$ items. It would be interesting to consider the mixed-size choice set where one may observe different numbers of items for each comparison. The Plackett-Luce model can be viewed as choosing the best item from $M$ items and then choosing the best from the remaining $M-1$ items and so on. For general mixed-size comparisons, some analyses need to be modified, and it is interesting to study whether the ranking inference results in this paper can be generalized. 
Secondly, although the convergence optimality and asymptotic normality results require no assumption on the number of comparisons $L$, when we studied the ranking inference in Section~\ref{sec:ranking_inference}, we needed $L$ to satisfy $L\gtrsim \textrm{poly}(\log n)$ in order to establish the theoretical guarantee for the Gaussian multiplier bootstrap. It remains open whether we can further relax the condition to $L\gtrsim 1$ or even allow $L=1$ to conduct effective ranking inferences. Thirdly, it would be interesting to see if 
some covariate information can be added to the analysis. In reality, the ranking is sometimes conducted together with item features or expert opinions. It is another exciting topic to study how we may incorporate these pieces of additional information into ranking inferences. 
Lastly, the time-varying effect of ranks may also be worth further investigation regarding time series of ranks. Over time, we may see underlying scores jump to a different level. Ranking inferences to detect the change point is another promising direction. Overall, we still see many challenges in inference for ranks under various settings, which calls for more research on ranking inference methodologies.

\begin{singlespace}
\small
\bibliographystyle{abbrvnat}
\bibliography{dynamic}

\end{singlespace}
\newpage 

\end{document}